\documentclass[aos,MSNbibl,nameyear,seceqn,dvips]{arximspdf}
\usepackage{graphicx}
\doi{10.1214/14-AOS1244} 
\volume{42}
\issue{5}
\pubyear{2014}
\firstpage{1911}
\lastpage{1940}
\docsubty{FLA}

\makeatletter
\def\mid{|}

\newcommand{\eqref}[1]{(\ref{#1})}

\newtheorem{theorem}{Theorem}[section]
\newtheorem{lemma}[theorem]{Lemma}
\newtheorem{corollary}[theorem]{Corollary}
\newtheorem{proposition}[theorem]{Proposition}
\newproclaim{assumption}[theorem]{Assumption}

\newproclaim{remark}[theorem]{Remark}
\newproclaim{example}[theorem]{Example}

\newcommand{\dnorm}{\varphi}
\newcommand{\pnorm}{\Phi}
\newcommand{\pr}{\mathbb{P}}
\newcommand{\prs}{\mathrm{P}}
\newcommand{\expec}{\mathbb{E}}
\newcommand{\expecs}{\mathrm{E}}
\newcommand{\var}{\operatorname{var}}
\newcommand{\cov}{\operatorname{cov}}
\newcommand{\Un}{\operatorname{Un}}
\newcommand{\reals}{\mathbb{R}}
\newcommand{\PLE}[1][]{\tilde\theta_{n #1}}
\newcommand{\OSE}[1][]{\hat\theta_{n #1}^{\mathrm{OSE}}}
\newcommand{\dell}{\dot{\ell}}
\newcommand{\ddell}{\ddot{\ell}}
\newcommand{\Score}{\dell_{\theta}}
\newcommand{\Scorem}{\dell_{\theta,m}}
\newcommand{\dto}{\mathop{\rightarrow}^{d}}
\newcommand{\diag}{\operatorname{diag}}
\newcommand{\tr}{\operatorname{tr}}
\newcommand{\ones}{\bolds{\iota}}
\newcommand{\Id}{I_p}
\newcommand{\hdot}{\circ}
\newcommand{\eps}{\varepsilon}
\newcommand{\Fac}{\mathcal{F}_{\mathrm{ac}}}
\newcommand{\rL}{\mathrm{L}}
\newcommand{\rd}{\mathrm{d}}
\newcommand{\modelpar}{\mathcal{P}_{\mathrm{km}}}
\newcommand{\model}{\mathcal{P}}
\newcommand{\scoreOpj}{\mathcal{O}_{\theta,j}}
\newcommand{\scoreOp}{\mathcal{O}_\theta}
\newcommand{\range}{\mathrm{R}}

\newcommand{\diagrem}{\mathrm{rd}}
\newcommand{\qdell}{q}

\newproclaim{definition}[theorem]{Definition}
\makeatother

\begin{document}
\begin{frontmatter}

\title{Semiparametric Gaussian copula models:
Geometry and efficient rank-based estimation}
\runtitle{Semiparametric Gaussian copula models}

\begin{aug}
\author[A]{\fnms{Johan}~\snm{Segers}\corref{}\ead[label=JS]{Johan.Segers@UCLouvain.be}\thanksref{t4}},
\author[B]{\fnms{Ramon}~\snm{van den Akker}\ead[label=RvdA]{R.vdnAkker@TilburgUniversity.edu}}
\and
\author[B]{\fnms{Bas~J.~M.}~\snm{Werker}\ead[label=BW]{B.J.M.Werker@TilburgUniversity.edu}}
\runauthor{J. Segers, R. van den Akker and B.~J.~M. Werker}
\affiliation{Universit\'{e} catholique de Louvain,
Tilburg University and
Tilburg University}
\thankstext{t4}{Supported by contract ``Projet d'Actions de Recherche
Concert\'ees'' No. 12/17-045 of the ``Communaut\'e fran\c{c}aise de
Belgique'' and by IAP research network Grant nr. P7/06 of the Belgian
government (Belgian Science Policy).}
\address[A]{J. Segers\\
ISBA\\
Universit\'e catholique de Louvain\\
Voie du Roman Pays 20\\
B-1348 Louvain-la-Neuve\\
Belgium\\
\printead{JS}} 
\address[B]{R. van den Akker\\
B.~J.~M. Werker\\
CentER\\
Tilburg University\\
PO Box 90153\\
5000 LE Tilburg\\
The Netherlands\\
\printead{RvdA}\\
\phantom{E-mail:\ }\printead*{BW}}
\end{aug}

\received{\smonth{6} \syear{2013}}
\revised{\smonth{5} \syear{2014}}

%
\begin{abstract}
We propose, for multivariate Gaussian copula models with unknown
margins and structured correlation matrices, a rank-based,
semiparametrically efficient estimator for the Euclidean copula
parameter. This estimator is defined as a one-step update of a
rank-based pilot estimator in the direction of the efficient influence
function, which is calculated explicitly.
Moreover, finite-dimensional algebraic conditions are given that
completely characterize efficiency of the pseudo-likelihood estimator
and adaptivity of the model with respect to the unknown marginal distributions.
For correlation matrices structured according to a factor model, the
pseudo-likelihood estimator turns out to be semiparametrically
efficient. On the other hand, for Toeplitz correlation matrices, the
asymptotic relative efficiency of the pseudo-likelihood estimator 
can be as low as 20\%.
These findings are confirmed by Monte Carlo simulations.
We indicate how our results can be extended to joint regression models.
\end{abstract}

%

\end{frontmatter}

\section{Introduction}
\label{sec:introduction}

Let $\mathbf{X} = (X_1, \ldots, X_p)^\prime$ be a $p$-dimensional random vector
with absolutely continuous marginal distribution functions $F_1,\ldots, F_p$ and joint distribution function $F$. The copula, $C$, of $F$ is
the joint distribution function of the vector $\mathbf{U} = (U_1,\ldots,
U_p)^\prime$ with $U_j = F_j(X_j)$, uniformly distributed on $(0, 1)$. By
Sklar's theorem,
\[
F(\mathbf{x}) = C\bigl(F_1(x_1),\ldots,
F_p(x_p)\bigr),\qquad \mathbf{x} \in\reals^p,
\]
yielding a separation of $F$ into its margins $F_1,\ldots, F_p$ and
its copula $C$. The copula remains unchanged if strictly increasing
transformations are applied to the $p$ components of $\mathbf{X}$.

A semiparametric copula model for the law of the random vector $\mathbf{X}$
is a model where $F$ is allowed to have arbitrary, absolutely
continuous margins 
and a copula $C_\theta$ which belongs to a finite-dimensional
parametric family.
An important inference problem is the
development of an efficient\vadjust{\goodbreak} estimator of the copula parameter $\theta$
on the basis of
a random sample $\mathbf{X}_i = (X_{i1},\ldots,X_{ip})^\prime$,
$i=1,\ldots,n$.
The marginal distributions $F_1,\ldots, F_p$ 
are thus considered as infinite-dimensional nuisance parameters.
In accordance to the structure of the model, a desirable property for
the estimator of $\theta$ is that it is invariant with respect to
strictly increasing transformations applied to the $p$ components. This
is equivalent to the requirement that the estimator is measurable with
respect to the vectors of ranks,\vspace*{-2pt} $\mathbf{R}_i^{(n)}=
(R_{i1}^{(n)},\ldots,R_{ip}^{(n)})^\prime$, with
$R_{ij}^{(n)}$\vspace*{2pt}  the rank of $X_{ij}$ within the $j$th
marginal sample $X_{1j},\ldots,X_{nj}$; see also \citeauthor{hoff:07}
[(\citeyear{hoff:07}), Section~5], who shows that the ranks are \textit{G}-sufficient.

There exist a number of rank-based estimation strategies for $\theta$,
none of them guaranteed to be semiparametrically efficient.
The most common 
are method-of-moment type estimators
[\citet{oakes:86,genest:rivest:93,kluppelberg:kuhn:09,brahimi:necir:12,liu:etal:2012}],
minimum-distance estimators [\citet{tsukahara:05,liebscher:09}], and
the pseudo-likelihood estimator [\citet{oakes:94,genest:ghoudi:rivest:95}].
For vine copulas, the pseudo-likelihood estimator and variants thereof
are studied in \citet{haff:13}. An expectation--maximization algorithm
for a Gaussian copula mixture model is proposed in \citet{li:etal:11}.

Conditions for the efficiency of the pseudo-likelihood estimator are
derived in \citet{genest:werker:02}, where it is concluded that
efficiency is the exception rather than the rule. One notable exception
is the bivariate Gaussian copula model (see below), where the
pseudo-likelihood estimator is asymptotically equivalent to the normal
scores rank correlation coefficient, shown to be efficient in \citet
{klaassen:wellner:97}.

A semiparametrically efficient estimator for the copula parameter is
proposed in \citet{chen:fan:tsy:06}. However, the estimator is based on
parametric sieves for the unknown margins, so that the estimator is not
invariant under increasing transformations of the component variables,
that is, the estimator is not rank-based. Moreover, it requires the
choice of the orders of the sieves as tuning parameters. The approach
has been extended to Markov processes [\citet{chen:wu:yi:2009}] and to
bivariate survival data [\citet{cheng:zhou:chen:huang:2013}].

Besides the already mentioned paper by \citet{klaassen:wellner:97}, the
issue of efficient, rank-based estimation is taken up in \citet{hoff:niu:wellner:12} for the important class of semiparametric
Gaussian copula models with structured correlation matrices.
They derive the semiparametric lower bound to the asymptotic variance of
rank-based regular estimators for the copula parameter. 
However, they do not provide such an estimator. They also construct a
specific Gaussian copula model for which the pseudo-likelihood
estimator is not efficient. They conclude their paper with the
suggestion that the maximum rank-likelihood estimator in \citet{hoff:07}
may be efficient.

Following \citet{klaassen:wellner:97} and \citet{hoff:niu:wellner:12}, we
put the focus in this paper on semiparametric Gaussian copula models.
In the context of graphical models for high-dimensional random vectors,
these have been called ``nonparanormal'' models, that is, the variables
follow a joint normal distribution after a set of unknown monotone
transformations
[\citet{liu:etal:2012,xue:zou:2012}].
Allowing for covariates, semiparametric Gaussian copula models also
lead to semiparametric regression models
[\citet{song:2000,song:li:yuan:2008,basrak:klaassen:2013}].

The Gaussian copula $C_\theta$ is the copula of a $p$-variate Gaussian
distribution $N_p(0, R(\theta))$ with $p \times p$ positive definite
correlation matrix $R(\theta)$:
%
\begin{equation}\qquad
\label{eq:Cgauss} C_\theta(\mathbf{u}) = \pnorm_\theta\bigl(
\Phi^{-1}(u_1),\ldots,\Phi^{-1}(u_p)
\bigr),\qquad \mathbf{u}=(u_1,\ldots,u_p)\in(0,1)^p,
\end{equation}
with $\pnorm_\theta$ the $N_p(0, R(\theta))$ cumulative distribution
function and $\Phi^{-1}$ the standard normal quantile function.
In the unrestricted Gaussian copula model, $R(\theta)$ can be any
$p\times p$ positive definite correlation matrix.
Submodels arise by considering structured correlation matrices. In that
case, the dimension, $k$, of the parameter set is smaller than the
number of pairs of variables, $p(p-1)/2$.
As model for one observation, we thus consider
%
\begin{equation}
\label{eq:model} \model= ( \prs_{\theta,F_1,\ldots,F_p} \mid \theta\in\Theta,
F_1,\ldots,F_p\in\Fac ),\qquad \Theta\subset
\reals^k,
\end{equation}
where $\Fac$ denotes the set of absolutely continuous distributions on
the real line and $\prs_{\theta,F_1,\ldots,F_p}$ denotes the law of the
random vector
$\mathbf{X}$ which has copula \eqref{eq:Cgauss} and margins $F_1,\ldots,F_p$.

For the model $\model$, we compute the efficient score and influence
functions and the efficient information matrix for
$\theta$.
The computations are based on a detailed analysis of the tangent space
structure of the model. The practical interest of this analysis is that
it yields an efficient semiparametric, rank-based estimator for $\theta
$. Starting from an initial, rank-based, $\sqrt{n}$-consistent
estimator, our estimator is defined as a one-step update in the
direction of the estimated efficient influence function.
The idea of constructing an efficient estimator via a one-step update
is already to be found in \citeauthor{bickel:1982}
[(\citeyear{bickel:1982}), Remark~1], who refers in
turn to \citeauthor{lecam:1969} [(\citeyear{lecam:1969}), Theorem~4].

Existence of an initial estimator is usually no problem: take the
pseudo-likelihood estimator, use a minimum-distance estimator, %
or express the parameter components as functions of the entries of the
correlation matrix and use a plug-in estimator based on the normal
scores correlation coefficients. In the semiparametric Gaussian copula
model, the update step is easy to implement, relying on simple matrix
algebra. 

We thereby provide a positive answer to the conjecture formulated in
\citet{hoff:niu:wellner:12} whether it is possible to attain the
semiparametric lower bound using a rank-based estimator. We wish to
stress the fact that, in contrast to the sieve-based estimator in \citet{chen:fan:tsy:06}, our estimator is rank-based, and thus invariant with
respect to increasing transformations of the component variables, in
accordance with the group structure of the model. Note, however, that
the methodology in \citet{chen:fan:tsy:06} also applies to general, not
necessarily Gaussian, copula models.

By restricting attention to estimators with influence functions of a
certain form, we construct an algebraic framework which allows for
particularly simple, finite-dimensional conditions on the
parametrization $\theta\mapsto R(\theta)$ for efficiency of the
pseudo-likelihood estimator 
and for adaptivity of the model. These are detailed for a large number
of examples. The pseudo-likelihood estimator turns out to be efficient
not only in the unrestricted model, confirming a remark in \citet
{klaassen:wellner:97}, but also in the often-used class of factor
models. On the other hand, for correlation matrices with a Toeplitz
structure, the pseudo-likelihood estimator can be quite inefficient,
with an asymptotic relative efficiency as low as 20\%. Although \citet{hoff:niu:wellner:12} already identified a Gaussian copula model for
which the pseudo-likelihood estimator is inefficient, the asymptotic
relative efficiency of the pseudo-likelihood estimator in their example
was still not far
from 100\%. The asymptotic results are complemented by a Monte Carlo
study, 
confirming the above theoretical findings for small samples, even in
high dimensions.


The outline of the paper is as follows. In Section~\ref{SEC:TANGENT},
we study the model's tangent space, culminating in the calculation of
the efficient score function and information matrix for $\theta$. These
serve to define the one-step estimator in Section~\ref{sec:estimator},
where its semiparametric efficiency is proved. Simple criteria for
efficiency of estimators and of adaptivity of the model are established
in Section~\ref{SEC:QUADFUN}. Examples and numerical illustrations are
provided in Section~\ref{SEC:EXAMPLES}. Section~\ref{sec:conclusions}
concludes and discusses extensions to other copula families and to
joint regression models. Detailed proofs are deferred to the
Appendices, which are collected in the supplement [\citet{supplement}].

\section{Tangent space and efficient score}
\label{SEC:TANGENT} 

The purpose of this section is to compute the semiparametric lower
bound for estimating the Gaussian copula parameter~$\theta$. Main keys
to obtain this lower bound are the tangent space of the semiparametric
Gaussian copula model $\model$ and the efficient score function for
$\theta$; see \citeauthor{bkrw:93} [(\citeyear{bkrw:93}), Chapters 2--3] and
\citeauthor{vdvaart:00} [(\citeyear
{vdvaart:00}), Chapter~25] for detailed expositions on these notions.
Readers mainly interested in our rank-based efficient estimator for
$\theta$ may want to jump to Section~\ref{sec:estimator} on a first reading.

Section~\ref{sec:assnot} states our assumptions and introduces notation
that will be used throughout. Section~\ref{sec:knownmarginals} shortly
discusses the Gaussian copula model with known marginals. In
Sections \ref{sec:tangentspace} and \ref{sec:effscore}, we determine
the tangent space and efficient score function, respectively. The
tangent space theory in Section~\ref{sec:tangentspace} is inspired upon
the one for bivariate semiparametric copula models in \citeauthor{bkrw:93}
[(\citeyear{bkrw:93}), Section~4.7].

\subsection{Assumptions and notation}
\label{sec:assnot}

The log density of the Gaussian copula \eqref{eq:Cgauss} with $p
\times
p$ correlation matrix $R(\theta)$ is given by
%
\begin{equation}
\label{eq:logdensity} \ell(\mathbf{u};\theta) = \log c_\theta(\mathbf{u}) = -
\tfrac{1}{2}\log\bigl( \det R(\theta) \bigr) - \tfrac{1}{2}
\mathbf{z}^\prime\bigl( R^{-1}(\theta) - \Id\bigr) \mathbf{z},
\end{equation}
for $\mathbf{u}\in(0,1)^p$, where $\Id$ is the $p\times p$ identity matrix
and where $\mathbf{z}=(z_1,\ldots,z_p)^\prime$ with $z_j = \Phi^{-1}(u_j)$.
Nonsingularity of the correlation matrix is part of the following assumption.

\begin{assumption}
\label{ass:parametrization}
Suppose $\Theta\subset\mathbb{R}^k$ is open and:
\begin{longlist}[(iii)]
\item[(i)]
the mapping $\theta\mapsto R(\theta)$ is one-to-one;
\item[(ii)]
for all $\theta\in\Theta$, the inverse $S(\theta)=R^{-1}(\theta)$ exists;
\item[(iii)]
for all $\theta\in\Theta$, the matrices of partial derivatives $\dot
R_1(\theta),\ldots,\dot R_k(\theta)$, defined by $\dot R_{ m, ij
}(\theta) =\partial
R_{ij}(\theta)/\partial\theta_m$, for $m=1,\ldots,k$ and
$i,j=1,\ldots,p$, exist and are continuous in $\theta$;
\item[(iv)]
for all $\theta\in\Theta$, the matrices $ \dot R_1(\theta),\ldots,\dot
R_k(\theta)$ are linearly independent.
\end{longlist}
\end{assumption}

Let us also define $p\times p$ matrices $\dot S_m(\theta)$ by $\dot
S_{ m, ij }(\theta) =\partial S_{ij}(\theta)/\partial\theta_m$.
These derivatives
satisfy $\dot S_m(\theta)= - S(\theta)   \dot R_m(\theta)
S(\theta
)$, which follows from differentiating $R(\theta)   S(\theta) = \Id$
[\citet{magnus:neudecker:99}, Section~8.4].

The $p$-dimensional vector $\mathbf{X}=(X_1,\ldots,X_p)^\prime$ denotes, as
in the \hyperref[sec:introduction]{Introduction}, a~random vector with copula (\ref{eq:Cgauss}) and
margins $F_1,\ldots,F_p\in\Fac$. Its law is denoted by $\prs_{\theta,F_1,\ldots,F_p}$ and expectations with respect to
this law are denoted by $\expecs_{\theta,F_1,\ldots,F_p}$.
In case all margins are uniform on $[0,1]$, notation $\Un[0, 1]$, we
use $\mathbf{U}$, $\prs_\theta$ and $\expecs_\theta$ as notation.
Let $\mathbf{Z}=(Z_1,\ldots,Z_p)^\prime$ be defined by $Z_j=\Phi
^{-1}(U_j)$ and
note that $\mathbf{Z}\sim\pnorm_\theta$ under $\prs_\theta$.

Moreover, we consider a measurable space $(\Omega,\mathcal{F})$
supporting probability measures $\pr_{\theta,F_1,\ldots,F_p}$, for
$\theta\in\Theta$ and $F_1,\ldots,F_p\in\Fac$, and i.i.d. random
vectors $\mathbf{X}_i$, $i\in\mathbb{N}$, each with law $\prs_{\theta,F_1,\ldots,F_p}$ under $\pr_{\theta,F_1,\ldots,F_p}$. Expectations with
respect to $\pr_{\theta,F_1,\ldots,F_p}$ are denoted by $\expec
_{\theta,F_1,\ldots,F_p}$. Furthermore, $\pr_\theta$ is shorthand for $\pr
_{\theta,\Un[0,1],\ldots,\Un[0,1]}$, for which expectations are denoted
by $\expec_\theta$.

\subsection{The Gaussian copula model with known margins}
\label{sec:knownmarginals}

The starting point of the analysis is the case that the margins $F_1,\ldots, F_p \in\Fac$ are known. In particular, we compute the Fisher
information matrix for $\theta$ in this case
. Due to the transformation structure of the model, it suffices to
consider uniform margins, that is, to consider the model $\modelpar
= ( \prs_\theta \mid  \theta\in\Theta )$.

Under Assumption~\ref{ass:parametrization}, the score
$\Score(\mathbf{u};\theta)=(\Scorem(\mathbf{u};\theta))_{m=1}^k$ is given by
%
\begin{equation}
\label{eqn:score} \Scorem(\mathbf{u};\theta) = \frac{\partial}{\partial\theta_m} \ell(\mathbf{u};
\theta) = -\frac{1}{2}\tr \bigl( S(\theta) \dot R_m(\theta)
\bigr) -\frac{1}{2} \mathbf{z}^\prime\dot S_m(\theta)
\mathbf{z},
\end{equation}
for $\mathbf{u}\in(0,1)^p$, where the partial derivative of $\det
R(\theta
)$ follows from Jacobi's formula [\citet{magnus:neudecker:99}, Section~8.3]. The $k\times k$ Fisher information matrix is
defined by $I_{}(\theta)=\expecs_\theta[\Score\Score^\prime(\mathbf
{U};\theta)]$.

To obtain a convenient representation of $I_{}(\theta)$, we introduce an
inner product on the linear space $\operatorname{Sym}(p)$ of real
symmetric $p\times p$
matrices. For $A,B\in\operatorname{Sym}(p)$, put
%
\begin{equation}
\label{eq:innerProductDef} \langle{A},{B}\rangle_\theta = \cov_\theta
\bigl(\tfrac{1}2\mathbf{Z}^\prime A\mathbf{Z},\tfrac{1}2
\mathbf {Z}^\prime B\mathbf {Z} \bigr) = \tfrac{1}2\tr \bigl( A R(
\theta) B R(\theta) \bigr),
\end{equation}
the covariance being calculated for $\mathbf{Z} \sim\pnorm_\theta$.
The latter equality follows from Theorem~12.10.12 in \citet
{magnus:neudecker:99}.
It is easily verified that $\langle{ \cdot },{ \cdot }\rangle
_\theta$ defines an
inner product on $\operatorname{Sym}(p)$. In particular, if $A\in
\operatorname{Sym}(p)$ is such that
$\langle{A},{A}\rangle_\theta=0$, then $\mathbf{Z}^\prime A\mathbf{Z}$ is
almost surely equal to
a constant and thus, $R(\theta)$ being nonsingular (Assumption~\ref
{ass:parametrization}), $A=0$.

From \eqref{eqn:score} and \eqref{eq:innerProductDef}, we obtain, for
$m,m^\prime=1,\ldots,k$,
%
\begin{equation}
\label{eqn:Fisher_parametric} I_{m m^\prime}(\theta) =\bigl\langle{-\dot S_m(
\theta)},{-\dot S_{m^\prime}(\theta)}\bigr\rangle _\theta,
\end{equation}
which is continuous in $\theta$. 
Note that $I_{}(\theta)$ is the Gram matrix associated to the matrices
$-\dot S_1(\theta),\ldots,-\dot S_k(\theta)$,\vspace*{1pt} using \eqref
{eq:innerProductDef} as inner product. Since Assumption~\ref
{ass:parametrization} implies linear independence of $-\dot S_1(\theta
),\ldots,-\dot S_k(\theta)$ (see part A of the proof of
Proposition~\ref
{prop:effscore} below), the information matrix $I_{}(\theta)$ is nonsingular.

From these observations, it follows that the Gaussian copula model with
known margins is regular [\citet{bkrw:93}, Definition~2.1.1 and
Proposition~2.1.1]. For ease of reference, we state this fact
in the following lemma.

\begin{lemma}
\label{prop:regular_parametric_model}
If the parametrization $\theta\mapsto R(\theta)$ satisfies
Assumption~\ref{ass:parametrization}, then the parametric Gaussian
copula model with known, uniform margins is regular.
\end{lemma}

\begin{remark}
\label{rem:regular_parametric_model}
Regularity of $\modelpar$ has some useful consequences:
\begin{longlist}[(ii)]
\item[(i)]
The score function has zero expectation, $\expecs_\theta[\Score(\mathbf
{U};\theta)] = 0$. This can also be seen from \eqref{eqn:score} and
$\expecs_\theta( \mathbf{Z}^\prime A \mathbf{Z}) = \tr( A   R(\theta))$%
.
\item[(ii)]
By the H\'{a}jek--Le Cam convolution theorem, the inverse of the Fisher
information matrix, $I^{-1}(\theta)$, constitutes a lower bound to the
asymptotic variance of regular estimators of $\theta$ in the model
$\modelpar$ [see, e.g., \citet{vdvaart:00}, Chapter~8].
\end{longlist}
\end{remark}

\subsection{Tangent space}
\label{sec:tangentspace}

In this section, we derive the tangent space of the semiparametric
Gaussian copula model $\model$. The structure of the space is similar
to the one of the bivariate semiparametric Gaussian copula model in
\citet{klaassen:wellner:97}. In Section~\ref{sec:effscore}, we use this
tangent space to calculate the efficient score for the copula parameter
$\theta$. In turn, the efficient score determines the semiparametric
lower bound for the asymptotic variance of regular estimators of
$\theta$.

Informally, the tangent space at $\prs_{\theta,F_1,\ldots,F_p}\in
\model$
is given by the collection of score functions of parametric submodels
of $\model$. Such score functions can be thought of as functions on
$\reals^p$ of the form
%
\begin{equation}
\label{eq:deriv:pointwise} \mathbf{x} \mapsto \frac{\partial}{\partial\eta} \log p_{\theta+\eta\alpha, F_{1,\eta},\ldots,F_{p,\eta}}(
\mathbf{x})\bigg|_{\eta=0}.
\end{equation}
Here, $\alpha\in\reals^k$, while $p_{\theta+\eta\alpha, F_{1,\eta
},\ldots,F_{p,\eta}}$ is the density of $\prs_{\theta+\eta\alpha,
F_{1,\eta},\ldots,F_{p,\eta}} \in\model$, depending on a real parameter
$\eta$ taking values in a neighborhood of $0$. The marginal
distributions are parametrized through paths $\eta\mapsto F_{j,\eta}$
in $\Fac$ that pass through $F_j$ at $\eta=0$.

The tangent space falls apart into two subspaces:
\begin{itemize}
\item[--]
a parametric part, arising from score functions of parametric submodels
for which the margins are constant, $F_{j,\eta} = F_j$;
\item[--]
a nonparametric part, arising from score functions of parametric
submodels for which the copula parameter is constant, $\alpha= 0$.
\end{itemize}
The parametric part corresponds in fact to the linear span of the score
functions in the parametric Gaussian copula model with known margins
. The nonparametric part describes the additional part of the model
stemming from the margins being unknown.

Formally, the tangent space is a subspace of $\rL_2(\prs_{\theta,F_1,\ldots,F_p})$, the space of square-integrable functions with
respect to $\prs_{\theta,F_1,\ldots,F_p}$. The pointwise derivatives
\eqref{eq:deriv:pointwise} will be replaced by derivatives in quadratic
mean. The nonparametric part of the tangent space is described most
conveniently as the image of a bounded linear operator, called score
operator. It is the description and the analysis of this score operator
that constitutes the gist of this section.

The density of $\prs_{\theta,F_1,\ldots,F_p} \in\model$ is given by
\[
p_{\theta, F_{1},\ldots,F_{p}}(\mathbf{x}) = c_{\theta} \bigl( F_{1}(x_1),
\ldots, F_{p}(x_p) \bigr) \prod
_{j=1}^p f_{j}(x_j),\qquad
\mathbf{x} \in\reals^p,
\]
with $f_1,\ldots, f_p$ the densities of $F_1,\ldots, F_p$,
respectively. If $\eta\mapsto F_{j,\eta}(x_j)$ is differentiable at
$\eta=0$, we obtain
\begin{eqnarray*}
&&\frac{\partial}{\partial\eta} \log p_{\theta, F_{1,\eta},\ldots,F_{p,\eta}}
(\mathbf{x}) \bigg|_{\eta=0}
\label{eqn:tangentspaceinformaldellj}
\\
&&\qquad=\sum_{j=1}^p \biggl\{ \frac{\partial}{\partial\eta}
\log f_{j,\eta}(x_j) \bigg|_{\eta=0} + \dot\ell_j
\bigl( F_{1}(x_1),\ldots, F_{p}(x_p);
\theta\bigr) \frac
{\partial
}{\partial\eta} F_{j,\eta}(x_j)\bigg|_{\eta=0}
\biggr\},
\end{eqnarray*}
with, for $j=1,\ldots,p$ and $\mathbf{u}\in(0,1)^p$,
%
\begin{equation}
\label{eqn:ellj} \dell_j(\mathbf{u};\theta) = \frac{\partial}{\partial u_j} \ell(
\mathbf{u};\theta) = \frac{z_j}{\dnorm(z_j)} - \sum_{i=1}^p
S_{ij}(\theta) \frac{z_i}{\dnorm(z_j)},
\end{equation}
where $\dnorm$ denotes the standard normal density.
Note that, for $u_j\in(0,1)$,
%
\begin{equation}
\label{eqn:condmean_dotellj} \expecs_\theta \bigl[ \dell_j(\mathbf{U};
\theta) | U_j=u_j \bigr] = \frac{z_j}{\dnorm(z_j)} \Biggl(1 -
\sum_{i=1}^p S_{ij}(
\theta)R_{ij}(\theta) \Biggr) = 0.
\end{equation}

The above formulas motivate the introduction of the following linear
operators, which together will constitute the above mentioned score
operator. Let $\rL_2^0[0,1]$ be the subspace of $\rL_2[0,1]=\rL_2(
[0,1],\mathcal{B}_{[0,1]},\rd\lambda)$ resulting from the restriction
$\int h(\lambda)   \,\rd\lambda= 0$ for $h \in\rL_2[0, 1]$.
For $j=1,\ldots,p$, we introduce linear operators $\scoreOpj: \rL
_2^0[0, 1] \to\rL_2(\prs_\theta)$ by
%
\begin{equation}
\label{eqn:scoreoperator} \scoreOpj h = [ \scoreOpj h ](\mathbf{U}) = h(U_j) +
\dot\ell_j (\mathbf{U};\theta) H(U_j), 
\end{equation}
where $H(u) = \int_0^{u} h(\lambda) \,\rd\lambda$ and where $\mathbf{U}$ is
the identity mapping on $(0, 1)^p$. The claim that the random variable
on the right-hand side has a finite variance for $\mathbf{U} \sim\prs
_\theta$ is part of Lemma~\ref{prop:score_operator}.

The score operator itself, $\scoreOp$, has domain $(\rL_2^0[0, 1])^p$
and is
defined by
\[
\scoreOp\mathbf{h} = \sum_{j=1}^p
\scoreOpj h_j, \qquad\mathbf{h} = (h_1,\ldots,h_p)
\in\bigl(\rL_2^0[0, 1]\bigr)^p.
\]
Lemma~\ref{prop:score_operator} will present basic properties of
$\scoreOpj$ and $\scoreOp$. A formal description of the tangent space
via the score operator will be given in Proposition~\ref
{prop:tangent_space_full}.

To this end, we first need to introduce some additional notation.
For $i,j=1,\ldots,p$ and $\mathbf{u}\in(0,1)^p$, we define
\begin{eqnarray*}
\ddell_{ij}(\mathbf{u};\theta) &=& \frac{\partial}{\partial u_i}
\dell_j(\mathbf{u};\theta)
\\
&=& \cases{\displaystyle \frac{z_j^2+1 -S_{jj}(\theta)}{\dnorm^2(z_j)} - \sum_{t=1}^p
S_{tj}(\theta) \frac{z_t z_j}{\dnorm^2(z_j)}, &\quad $\mbox{if $i=j$;}$ \vspace*{2pt}
\cr
\displaystyle -\frac{S_{ij}(\theta)}{\dnorm(z_i)   \dnorm(z_j)}, & \quad$\mbox{if $i \ne j$.}$}
\end{eqnarray*}
For $u_j\in(0,1)$, we have
%
\begin{eqnarray}
\label{eqn:Ijj} I_{jj}(u_j;\theta) &=&
\expecs_\theta \bigl[ \dell_j^2(\mathbf{U};
\theta) | U_j=u_j \bigr]
\nonumber
\\[-8pt]
\\[-8pt]
\nonumber
&=& -\expecs_\theta \bigl[ \ddell_{jj}(\mathbf{U};\theta) |
U_j=u_j \bigr] = \frac{ S_{jj}(\theta)-1}{\dnorm^2(z_j)}.
\end{eqnarray}
From well-known results on Mill's ratio [\citet{gordon:41}], we obtain
the bound $ 1 / \dnorm(z_j) \leq M \{ u_j(1-u_j) \}^{-1}$, for all $u_j
\in(0, 1)$ and some constant $M > 0$. Hence, under Assumption~\ref
{ass:parametrization}, there exists a constant $M_\theta>0$ such that
%
\begin{equation}
\label{eqn:bound_Ijj} I_{jj}(u_j;\theta)
 \leq\frac{M_\theta}{\{ u_j (1-u_j) \}^{2}},\qquad
u_j \in(0, 1).
\end{equation}
This bound is exploited in the proof of the following lemma, which
states that $\scoreOp$ is a continuously invertible operator
from $(\rL_2^0[0, 1])^p$ into $\rL_2^0(\prs_\theta)$. Here we equip
$(\rL_2^0[0, 1])^p$ with the
inner product $\langle\mathbf{g}, \mathbf{h} \rangle= \sum_{j=1}^p \int_0^1
g_j(\lambda)   h_j(\lambda)   \,\rd\lambda$ for $\mathbf{g}, \mathbf{h}
\in(\rL_2^0[0, 1])^p
$, while
$\rL_2^0(\prs_\theta)$ is the subspace of $\rL_2(\prs_\theta)$
resulting from the restriction $\expecs_\theta[ f(\mathbf{U}) ] = 0$ for
$f \in\rL_2(\prs_\theta)$.

\begin{lemma}
\label{prop:score_operator}
Let $\theta\mapsto R(\theta)$ be a parametrization that satisfies
Assumption~\ref{ass:parametrization} and let $\theta\in\Theta$:
%
\begin{longlist}[(a)]
\item[(a)]
The map $\scoreOpj$, $j = 1,\ldots, p$, is a bounded operator from
$\rL
_2^0[0,1]$ into $\rL_2^0(\prs_\theta)$. The map $\scoreOp$ is a bounded
operator from $(\rL_2^0[0, 1])^p$ into $\rL_2^0(\prs_\theta)$.
\item[(b)]
The operator $\scoreOp$ is continuously invertible on its range,
$\range\scoreOp$, with inverse $\scoreOp^{-1} \dvtx \rL_2^0( \prs
_\theta)
\to(\rL_2^0[0, 1])^p$ given by
\[
\bigl[\scoreOp^{-1} f\bigr]_j(u_j) =
\expecs_\theta\bigl[ f( \mathbf{U} ) \mid U_j =
u_j \bigr],\qquad u_j \in(0, 1), j = 1,\ldots, p.
\]
%
\end{longlist}
\end{lemma}

The proof is given in Appendix A
in the supplement. Here, we just note that the proof of part (a)
follows the one of Proposition~4.7.2 in \citet{bkrw:93}.

\begin{remark}
\label{rem:score_operator}
An application of Banach's theorem [see,
e.g., \citet{bkrw:93}, Proposition A.1.7] implies that $\range\scoreOp$ is closed.
\end{remark}

We proceed with the construction of the tangent space. First, we define
the ``local paths'' through $\Fac$. To this end, fix $F_1,\ldots,F_p\in
\Fac$ with densities $f_1,\ldots,f_p$. Let $h_1,\ldots,h_p \in\rL
^0_{2}[0,1]$ and introduce univariate densities $f_{j,\eta}(   \cdot; h_j)$, for $\eta\in(-1,1)$ and $j = 1,\ldots, p$, by
\[
f_{j,\eta}(x; h_j) = d(\eta; f_j,
h_j) g\bigl(\eta h_j\bigl(F_j(x)\bigr)
\bigr) f_j(x), \qquad x \in\reals,
\]
where $g(z)= 2 / ( 1 + e^{-2z} )$ and where $d(\eta;f_j,h_j)$ is a
positive constant such that $f_{j,\eta}(\cdot;h_j)$ integrates to one.
Let $F_{j,\eta}^{h_j}\in\Fac$ be the induced distribution functions.
The path $\eta\mapsto(F_{1,\eta}^{h_1},\ldots,F_{p,\eta}^{h_p})$, with
values in the space $(\Fac)^p$, passes through $(F_1,\ldots,F_p)$ at
$\eta= 0$.

The following proposition describes the score function at $\eta= 0$ of
the parametric submodel
\[
( \prs_{
\theta+ \eta\alpha,
F_{1,\eta}^{h_1},\ldots, F_{p,\eta}^{h_p}
} | -\eps< \eta< \eps ) \subset\model
\]
for fixed $\alpha\in\reals^k$ and $h_1,\ldots, h_p \in\rL_2^0[0,
1]$ and some $\eps> 0$. The collection of all such score functions is,
by definition, the tangent set of the semiparametric Gaussian copula
model $\model$ at $\prs_{\theta,F_1,\ldots,F_p}$.

\begin{proposition}
\label{prop:tangent_space_full}
Consider a parametrization $\theta\mapsto R(\theta)$ for which
Assumption~\ref{ass:parametrization} holds and let $\prs_{\theta,F_1,\ldots,F_p}\in\model$. Let $\mathbf{h}=(h_1,\ldots,h_p) \in(\rL
_2^0[0, 1])^p$
and let
$\alpha\in\reals^k$.
Then the path $\eta\mapsto
(\theta+ \eta\alpha, F_{1,\eta}^{h_1},\ldots,F_{p,\eta}^{h_p})$, in
$\Theta\times( \Fac)^p$, yields the following score at $\eta=0$:
\[
\dot\ell^{\alpha,\mathbf{h}}(\mathbf{x}) = \alpha^\prime\Score
\bigl(F_1(x_1),\ldots,F_p(x_p);
\theta\bigr)+ [\scoreOp\mathbf{h} ]\bigl(F_1(x_1),\ldots,
F_p(x_p)\bigr)
\]
for $\mathbf{x} \in\reals^p$, that is,
\[
\lim_{\eta\to0} \int_{\reals^p} \biggl(
\frac{\sqrt{p_{\eta}(\mathbf{x})} - \sqrt{p_{0}(\mathbf{x})}}{\eta} -\frac{1}{2} \dot\ell^{\alpha,\mathbf{h}}(\mathbf{x})
\sqrt{p_{0}(\mathbf{x})} \biggr)^2 \,\rd\mathbf{x} = 0,
\]
where 
$p_\eta= p_{
\theta+ \eta\alpha,
F_{1,\eta}^{h_1},\ldots, F_{p,\eta}^{h_p}
}$. The tangent set
%
\begin{equation}
\label{eqn:fulltangentspace} \bigl\{ \dot\ell^{\alpha,\mathbf{h}} | \alpha\in\reals^k,
\mathbf{h} \in\bigl(\rL_2^0[0, 1]\bigr)^p
\bigr\},
\end{equation}
is a closed subspace of $\rL_2^0(\prs_{\theta,F_1,\ldots,F_p})$.
\end{proposition}

The tangent set \eqref{eqn:fulltangentspace}, being a closed subspace,
is called the \emph{tangent space} of $\model$ at $\prs_{\theta,
F_1,\ldots, F_p}$. Apart from the statement on closedness, which we
discuss below, the proof of Proposition~\ref{prop:tangent_space_full}
is analogous to the proof of Proposition~4.7.4 in \citet{bkrw:93} and
is omitted.

The nonparametric part $\mathcal{T}_{\prs_{\theta,F_1,\ldots,F_p}}$
of the tangent space at $\prs
_{\theta,F_1,\ldots,F_p}$ corresponds, by definition, to the subset of
\eqref{eqn:fulltangentspace} resulting from the restriction $\alpha=0$,
that is,
%
\begin{eqnarray}
\label{eq:tangentspace}
\mathcal{T}_{\prs_{\theta,F_1,\ldots,F_p}} &=& \bigl\{ \dot\ell^{0, \mathbf{h}} |
\mathbf{h} \in\bigl(\rL _2^0[0, 1]\bigr)^p
\bigr\}
\nonumber
\\[-8pt]
\\[-8pt]
\nonumber
 &= &\bigl\{ \mathbf{x} \mapsto[\scoreOp\mathbf{h} ]
\bigl(F_1(x_1),\ldots,F_p(x_p)
\bigr) | \mathbf{h}\in\bigl(\rL_2^0[0, 1]
\bigr)^p \bigr\}.
\end{eqnarray}
Since $\mathcal{T}_{\prs_{\theta,F_1,\ldots,F_p}}$ is isometric to
$\mathcal{T}_{\prs_{\theta}}$
and since $\mathcal{T}_{\prs_{\theta}}=\range\scoreOp$, which is
closed, it
follows that
$\mathcal{T}_{\prs_{\theta,F_1,\ldots,F_p}}$ is a closed subspace of
$\rL_2^0(\prs_{\theta,F_1,\ldots,F_p})$. The tangent space~\eqref{eqn:fulltangentspace}, being the sum
of $\mathcal{T}_{\prs_{\theta,F_1,\ldots,F_p}}$ and a
finite-dimensional space, is closed as well.

\subsection{Efficient score}
\label{sec:effscore}

The semiparametric lower bound for regular estimators of the copula
parameter $\theta$ 
is determined by the efficient score, $\dell_{\theta}^*( \mathbf{X};
\prs
_{\theta,F_1,\ldots,F_p})$. (As before, we identify square-integrable functions
with random variables; formally, view $\mathbf{X}$ as the identity map on
$\reals^p$.) This efficient score is, by definition, given by
\begin{eqnarray*}
\dell_{\theta}^*( \mathbf{X}; \prs_{\theta,F_1,\ldots,F_p}) &=& \dot
\ell_\theta\bigl(F_1(X_1),\ldots,F_p(X_p);
\theta\bigr)
\\
&&{}-\Pi\bigl( \Score\bigl(F_1(X_1),\ldots,F_p(X_p);
\theta\bigr) \mid\mathcal {T}_{\prs_{\theta,F_1,\ldots,F_p}}\bigr),
\end{eqnarray*}
where $\Pi(  \cdot  | \mathcal{T}_{\prs_{\theta,F_1,\ldots,F_p}}) $ is the (coordinate-wise)
projection operator from\break $\rL_2(\prs_{\theta,F_1,\ldots,F_p})$ onto the
closed subspace $\mathcal{T}_{\prs_{\theta,F_1,\ldots,F_p}}$. Note
that $\dot\ell_\theta$, and hence
$\dell_{\theta}^*$ are vectors of length $k$, the length of the copula
parameter $\theta$.

For Gaussian copula models,
parametrized by $\theta\mapsto R(\theta)$, the
projection can be calculated explicitly, which will lead eventually to
our one-step estimator in Section~\ref{sec:estimator}. This situation
is in contrast to the one for most other copula models, even for
bivariate copulas indexed by a real-valued parameter, where the
calculation of the efficient score requires the solution of a pair of
coupled Sturm--Liouville differential equations [\citet{bkrw:93}, Section~4.7].

From the isometry between $\mathcal{T}_{\prs_{\theta,F_1,\ldots,F_p}}$ and $\mathcal{T}_{\prs_{\theta}}$, via the
mapping $\mathbf{X} \mapsto(F_1(X_1), \break \ldots, F_p(X_p))$, we obtain the
important identity
%
\begin{equation}
\label{eqn:effscore_reduction}\qquad \dell_{\theta}^*(\mathbf{X}; \prs_{\theta,F_1,\ldots,F_p}) =
\dell_{\theta}^*\bigl(F_1(X_1),\ldots,F_p(X_p);
\theta\bigr), \qquad\prs_{\theta,F_1,\ldots,F_p}\mbox{-a.s.},
\end{equation}
where $\dell_{\theta}^*(  \cdot; \theta)$ is shorthand notation for
$\dell_{\theta}^*(  \cdot; \prs_{\theta} )$ and $\prs_\theta$
is shorthand
for $\prs_{\theta,\Un[0,1],\ldots,\Un[0,1]}$. As a consequence, the
$k\times k$ efficient information matrix for $\theta$ at $\prs
_{\theta,F_1,\ldots,F_p}$, given by
\[
I^*_{}(\theta) = \expecs_{\theta,F_1,\ldots,F_p} \bigl[ \dell_{\theta}^*
\dell_{\theta}^{*\prime} ( \mathbf{X}| \prs_{\theta,F_1,\ldots,F_p}) \bigr] =
\expecs_{\theta} \bigl[ \dell_{\theta}^*\dell_{\theta}^{*\prime}(
\mathbf{U}; \theta) \bigr],
\]
does not depend on the marginals $F_1,\ldots,F_p$.

Because $\mathcal{T}_{\prs_{\theta}}=\range\scoreOp$ and $\scoreOp
$ is one-to-one
(see Proposition~\ref{prop:score_operator}), there exist unique
elements $\mathbf{h}^{\theta}_{m} = (h^{\theta}_{1,m}, \ldots, h^{\theta
}_{p,m}) \in(\rL_2^0[0, 1])^p$, $m=1,\ldots,k$, such that
%
\begin{equation}
\label{eq:hthetam} \Pi\bigl( \Scorem(\mathbf{U};\theta) \mid\mathcal{T}_{\prs_{\theta}}
\bigr) = \bigl[\scoreOp\mathbf{h}^{\theta}_{m} \bigr](
\mathbf{U}).
\end{equation}
These ``generators of the efficient score'' are completely determined
by the orthogonality conditions
%
\begin{equation}\qquad
\label{eqn:orthocond} 0=\expecs_\theta \bigl[ \bigl( \Scorem(\mathbf{U};
\theta)- \bigl[\scoreOp\mathbf{h}^\theta_m\bigr](\mathbf{U})
\bigr) [\scoreOp\mathbf{h}](\mathbf{U}) \bigr], \qquad\mathbf{h}\in\bigl(
\rL_2^0[0, 1]\bigr)^p.
\end{equation}
Before we solve 
for $\mathbf{h}_m^\theta$, $m=1,\ldots,k$, we provide a necessary and
sufficient condition for orthogonality of quadratic forms in the
Gaussianized variables to the space $\mathcal{T}_{\prs_{\theta}}$.
Its proof is
given in Appendix A 
in the supplement.

\begin{lemma}
\label{prop:orthogonalitytangentspace}
Consider a parametrization $\theta\mapsto R(\theta)$ for which
Assumption~\ref{ass:parametrization} holds and let $\theta\in\Theta$
and $A \in\operatorname{Sym}(p)$. The following two conditions are equivalent:
\begin{longlist}[(a)]
\item[(a)] the function $\mathbf{Z}' A \mathbf{Z}$ in $\rL_2(\prs_\theta)$ is
orthogonal to $\mathcal{T}_{\prs_{\theta}}$;
\item[(b)]$(R(\theta) A)_{jj} = 0$ for $j = 1,\ldots, p$.
\end{longlist}
\end{lemma}


The following proposition presents the solution $\mathbf{h}^\theta_m$
to \eqref{eqn:orthocond} and the resulting efficient score and
efficient information matrix. To formulate these results, we introduce
the notation, for $\mathbf{b} \in\reals^p$ and $\theta\in\Theta$,
%
\begin{equation}
\label{eqn:Dmatrix} D_\theta(\mathbf{b}) = S(\theta) \diag(\mathbf{b}) +
\diag(\mathbf{b}) S(\theta),
\end{equation}
where $\diag(\mathbf{b})$ is the diagonal matrix with diagonal $\mathbf{b}$.
Let $\ones_p$ denote the $p$-dimensional vector of ones and let
$A\hdot
B$ denote the Hadamard product of conformable matrices $A$ and $B$.
Recall the inner product 
introduced in \eqref{eq:innerProductDef} and recall the convention $z =
\Phi^{-1}(u)$ for $u \in(0, 1)$.

\begin{proposition}
\label{prop:effscore}
Consider a parametrization $\theta\mapsto R(\theta)$ for which
Assumption~\ref{ass:parametrization} holds and let $\theta\in\Theta$.
Then, for $m=1,\ldots,k$, the vector $\mathbf{h}_m^\theta\in(\rL_2^0[0,
1])^p$ in
\eqref
{eq:hthetam} is given by
%
\begin{equation}
\label{eqn:effscore_generator} h_{j,m}^\theta(u_j) =
g_{j,m}(\theta) \bigl( 1-z_j^2\bigr),\qquad
u_j\in(0,1), j=1,\ldots,p,
\end{equation}
where $\mathbf{g}_m(\theta)=(g_{1,m}(\theta),\ldots,g_{p,m}(\theta
))^\prime$
is given by
%
\begin{equation}
\label{eqn:effscore_generator_weights} \mathbf{g}_{m}(\theta) = - \bigl( I_p+R(
\theta)\hdot S(\theta) \bigr)^{-1} \bigl( \dot R_m(\theta)
\hdot S(\theta) \bigr) \ones_p.
\end{equation}
Moreover, the efficient score, $\dell_{\theta}^*$, is given by, for
$m=1,\ldots,k$,
%
\begin{equation}
\label{eqn:effscorem} \dell_{\theta,m}^*(\mathbf{u};\theta) = \tfrac{1}{2}
\mathbf{z}^\prime \bigl( D_\theta\bigl(\mathbf{g}_m(
\theta)\bigr) -\dot S_m(\theta) \bigr) \mathbf{z}.
\end{equation}
Finally, the efficient information matrix, $I^*(\theta)$, is given by
%
\begin{equation}
\label{eqn:effInfoMat} I^*_{mm^\prime}(\theta) = \bigl\langle{D_\theta\bigl(
\mathbf{g}_{m}(\theta)\bigr)-\dot S_{m}(\theta
)},{D_\theta\bigl(\mathbf{g}_{m^\prime}(\theta)\bigr)-\dot
S_{m^\prime}(\theta )}\bigr\rangle_\theta,
\end{equation}
for $m,m^\prime=1,\ldots,k$, and is nonsingular.
\end{proposition}

\begin{pf*}{Outline of the proof}
The proof is decomposed into four parts. Here, we just give an outline.
For a detailed proof, please see Appendix A 
in the supplement.

In part A, we show that the $k+p$ matrices
\[
-\dot S_1(\theta), \ldots, -\dot S_k(\theta),
D_\theta(\mathbf{e}_1), \ldots, D_\theta(
\mathbf{e}_p),
\]
with $\mathbf{e}_i$ the $i$th canonical unit vector in $\reals^p$, are
linearly independent. Part B exploits this result to demonstrate
nonsingularity of $I_p + R(\theta) \hdot S(\theta)$, thereby showing
that the vector $\mathbf{g}_m(\theta)$ in \eqref
{eqn:effscore_generator_weights} is well defined.
Part C shows that equation \eqref{eqn:effscore_generator} together with
the definition
%
\begin{equation}
\label{eqn:partC} \dell_{\theta,m}^*(\mathbf{u}; \theta) = \dot
\ell_{\theta,m}(\mathbf{u}; \theta) - \bigl[ \scoreOp\mathbf{h}_m^\theta
\bigr]( \mathbf{u} )
\end{equation}
lead to \eqref{eqn:effscorem}--\eqref{eqn:effInfoMat}
and also demonstrates nonsingularity of $I^*_{}(\theta)$. Finally,
part~D proves that the orthogonality conditions \eqref{eqn:orthocond}
hold by applying Lemma~\ref{prop:orthogonalitytangentspace}, and thus
shows that \eqref{eqn:effscorem} is the efficient score.
\end{pf*}

\begin{remark}
The $k\times k$ positive semidefinite matrix $I_{}(\theta
)-I^*_{}(\theta)$
represents the loss of information due to not knowing the marginals.
In Section~\ref{subsec:adaptivity}, we provide conditions for
adaptivity, that is, $I_{}(\theta) = I^*_{}(\theta)$.
\end{remark}

\begin{remark}
\label{rem:score}
\citet{klaassen:wellner:97} derived the efficient score for the
bivariate ($p=2$) unrestricted Gaussian copula model by solving a
system of Sturm--Liouville equations. The functions \eqref
{eqn:effscore_generator} solve a $p$-variate analogue for the component
$m\in\{1,\ldots,k\}$ of the efficient score; see equation (A.11)
in the supplement.
\end{remark}

\section{An efficient rank-based estimator}
\label{sec:estimator}

In this section, we use the efficient score $\dell_{\theta}^*$ and the
efficient information matrix $I^*(\theta)$, as obtained in
Proposition~\ref{prop:effscore}, to construct a rank-based,
semiparametrically efficient estimator of $\theta$. Recall [\citet
{vdvaart:00}, Sections 25.3--25.4] that $\hat{\theta}_n$ is an
efficient estimator of $\theta$ in model $\model$ at $\prs_{\theta,F_1,\ldots,F_p}\in\model$ if and only if, under $\pr_{\theta,F_1,\ldots,F_p}$,
%
\begin{equation}
\label{eqn:efficiency} \sqrt{n}( \hat{\theta}_n - \theta) =
\frac{1}{\sqrt{n}} \sum_{i=1}^n
I^{*-1}(\theta){} \dell_{\theta}^*\bigl( F_1(X_{i1}),
\ldots, F_p(X_{ip}); \theta\bigr) + o_P(1).
\end{equation}
Moreover, $\hat{\theta}_n$ is called efficient (in the model $\model$)
if it is efficient at all $\prs_{\theta, F_1,\ldots, F_p} \in
\model$.
The limiting distribution of an efficient estimator is thus given by
$N_k(0, \break I^{*-1}(\theta){})$. By the H\'{a}jek--Le Cam convolution
theorem, the limiting distribution of any regular estimator is given by
the convolution of $N_k(0, I^{*-1}(\theta){})$ and another, estimator
specific, distribution. As a consequence, $I^{*-1}(\theta){}$ provides a
lower bound to the asymptotic variance of regular estimators. See
Section~\ref{sec:regularity} for an insightful characterization of
regularity for estimators in structured Gaussian copula models.

The vector-valued function
%
\begin{equation}
\label{eq:effInfFun} 
\mathbf{x} \mapsto I^{*-1}(\theta){}
\dell_{\theta}^*\bigl( F_1(x_{1}),\ldots,
F_p(x_{p}); \theta\bigr)
\end{equation}
is called the \emph{efficient influence function}. According to \eqref
{eqn:efficiency}, an estimator sequence $\hat{\theta}_n$ is efficient
for $\theta$ if and only if $\sqrt{n} (\hat{\theta}_n - \theta)$ is
asymptotically linear in the efficient influence function. By
Proposition~\ref{prop:effscore}, each component of the efficient
influence function is a centered quadratic form in the Gaussianized
vector $\mathbf{z} \in\reals^p$, where $z_j = \Phi^{-1}(u_j)$ and $u_j =
F_j(x_j)$. This fact will be extensively used in Section~\ref{SEC:QUADFUN}.

We construct an efficient one-step estimator (OSE) by updating an
initial $\sqrt{n}$-consistent estimator of $\theta$. Since we want to
construct a rank-based estimator of~$\theta$, we require that the
initial estimator is rank-based, too. We summarize these requirements
in the following assumption. Recall that we consider a measurable space
equipped with probability measures $\pr_{\theta,F_1,\ldots,F_p}$ and
carrying random vectors $\mathbf{X}_i = (X_{i1},\ldots, X_{ip})'$, for $i
\ge1$, which, under $\pr_{\theta,F_1,\ldots,F_p}$, are i.i.d. with
common law $\prs_{\theta, F_1,\ldots, F_p} \in\model$. Also recall
that $\mathbf{R}_i^{(n)}= (R_{i1}^{
(n)},\ldots, R_{ip}^{(n)})^\prime$ denotes a
vector of ranks, with $R_{ij}^{(n)}$ the rank of
$X_{ij}$ within the $j$th
marginal sample $X_{1j},\ldots,X_{nj}$.

\begin{assumption}
\label{ass:rootnconsistent}
There exists an estimator $\PLE^* = t_n(\mathbf{R}_1^{(n)}, \ldots,\mathbf
{R}_n^{(n)})$ 
such that, for all $\prs_{\theta,F_1,\ldots,F_p}\in\model$, we have
$\sqrt{n}(\PLE^*-\theta)=O_P(1)$ under $\pr_{\theta,F_1,\ldots,F_p}$.
\end{assumption}

An obvious candidate 
is the pseudo-likelihood estimator. 
If the copula parameter $\theta$ can be expressed as a smooth
function of the correlation matrix $R(\theta)$, an alternative is to
construct a minimum distance type estimator of $\theta$ using the
normal scores rank correlations.

In what follows, $\PLE$ denotes a discretized version of $\PLE^*$,
obtained by rounding 
$\PLE^*$ to the grid $n^{-1/2}\mathbb{Z}^k$. Discretization of the
initial estimator is needed in the efficiency proof but has little to
no practical implications [see pages~125 or 188 in \citet{LeCamYang:90}
for a discussion].

From Proposition~\ref{prop:effscore}, recall the efficient score
function $\dell_{\theta}^*$ and the efficient information matrix
$I^*_{}(\theta)$. Further, rescaled versions of the marginal empirical
distribution
functions are provided by
\[
\hat F_{n,j}(x) = \frac{1}{n+1}\sum_{i=1}^n
1\{ X_{ij}\leq x \},\qquad x \in\reals, j = 1,\ldots, p.
\]

The \emph{one-step estimator} is then defined by
%
\begin{equation}
\label{eqn:OSE} \OSE= \PLE+ \frac{1}{n} \sum_{i=1}^n
I^{*-1}(\PLE){} \dell_{\theta}^* \bigl( \hat F_{n,1}(X_{i1}),
\ldots, \hat F_{n,p}(X_{ip}); \PLE \bigr).
\end{equation}
The initial estimator being rank-based, the one-step estimator is
rank-based, too. In particular, $\OSE$ is invariant with respect to
strictly increasing transformations applied to each of the $p$ variables.

\citet{hoff:niu:wellner:12} raised the question whether it is possible
to develop a rank-based, semiparametrically efficient estimator for
Gaussian copula models.
The following theorem gives a positive answer: the proposed one-step
estimator is efficient.

\begin{theorem}
\label{THM:EFFICIENCY} 
Suppose that the parametrization $\theta\mapsto R(\theta)$ satisfies
Assumption~\ref{ass:parametrization}. Also assume that the initial
estimator $\PLE^*$ satisfies Assumption~\ref{ass:rootnconsistent}. Then
$\OSE$ 
is an efficient, rank-based estimator of $\theta$ in the semiparametric
Gaussian copula model $\model$, that is, for all $F_1, \ldots, F_p \in
\Fac$ and $\theta\in\Theta$ we have \eqref{eqn:efficiency} with
$\hat
{\theta}_n$ replaced by $\OSE$ and thus, as $n \to\infty$,
\[
\sqrt{n} \bigl(\hat\theta_{n}^{\mathrm{OSE}} - \theta\bigr) \dto
N_k\bigl(0, I^{*-1}(\theta){}\bigr).
\]
\end{theorem}

\begin{pf*}{Outline of the proof}
We give a detailed proof in Appendix B 
in the supplement and present here a short outline. First, we show that
it suffices to prove the theorem for uniform marginals. Let $\mathbf{U}_i$,
$i \in\mathbb{N}$, be i.i.d. random vectors with law $\prs_\theta$
under $\pr_\theta$. Following the lines of the proof of the efficiency
of \emph{parametric} one-step estimators [\citet
{bkrw:93}, Theorem~2.5.2], we show that \eqref{eqn:efficiency} holds
if:
\begin{longlist}[(P1)]
\item[(P1)]
for any sequence $\theta_n=\theta+ h_n/\sqrt{n}$, with $h_n \in
\reals
^k$ bounded, we have
\[
\frac{1}{\sqrt{n}}\sum_{i=1}^n
\dell_{\theta}^*(\mathbf{U}_i;\theta_n) =
\frac{1}{\sqrt{n}}\sum_{i=1}^n
\dell_{\theta}^*(\mathbf{U}_i;\theta ) - I^*_{}(
\theta) h_n + o(1;\pr_\theta);
\]
\item[(P2)]
for any sequence $\theta_n=\theta+ h_n/\sqrt{n}$, with $h_n \in
\reals
^k$ bounded, we have
\[
\frac{1}{\sqrt{n}} \sum_{i=1}^n
\dell_{\theta}^* \bigl( \hat F_{n,1}(U_{i1}),\ldots,
\hat F_{n,p}(U_{ip}); \theta_n \bigr) =
\frac{1}{\sqrt{n}} \sum_{i=1}^n
\dell_{\theta}^*( \mathbf{U}_i;\theta_n ) + o(1;
\pr_\theta).
\]
\end{longlist}
We show that (P1) holds by exploiting the smoothness of the efficient
score and Proposition A.10 in \citet{vdvaart:88}.
Statement (P2) follows from a modification of Corollary~3.1 in \citet{hoff:niu:wellner:12}.
\end{pf*}

\begin{remark}
\label{rem:discr}
(i) The special structure \eqref{eqn:effscore_reduction} of the efficient
score allows us to follow the lines of the proof of the efficiency of
\emph{parametric} one-step estimators. We thus do not need to use
sample-splitting as in, for example, \citet{klaassen:87}.\vspace*{-6pt}
%
\begin{longlist}[(ii)]
\item[(ii)]
The update step \eqref{eqn:OSE} is simple to implement in practice.
Regarding the calculation of the efficient information matrix in \eqref
{eqn:effInfoMat}, recall that the inner product $\langle{\cdot
},{\cdot}\rangle_\theta$
is defined in \eqref{eq:innerProductDef} and that $\dot{S}_m(\theta) =
- S(\theta)   \dot{R}_m(\theta)   S(\theta)$.
\end{longlist}
\end{remark}

\section{Quadratic influence functions}
\label{SEC:QUADFUN} 


Structured Gaussian copula models are completely specified by the
parametrization $\theta\mapsto R(\theta)$. This implies that
conditions for, for example, efficiency of a given estimator or
adaptivity of the model at a certain value of $\theta$ can be expressed
in terms of this parametrization. In addition, all relevant score and
influence functions turn out to be quadratic in the Gaussianized
observations in the sense of \eqref{eq:quadInfFunc}. In this section,
we will give finite-dimensional algebraic conditions under which an
estimator with such an influence function is regular (Proposition~\ref
{prop:quadCondRegularity}) or even efficient (Proposition~\ref
{prop:quadCondEfficiency}). In addition, Proposition~\ref
{prop:quadCondAdaptivity} gives a simple condition for adaptivity of
the model at some value of $\theta$.

For practice, the most relevant result in this section concerns the
characterization of those Gaussian copula models for which the
pseudo-likelihood estimator, $\hat{\theta}_n^{
\mathrm{PLE}}$, is efficient.
Recall, see \citet{genest:ghoudi:rivest:95}, that $\hat{\theta
}_n^{\mathrm{PLE}}$
is the maximizer of
$
\theta\mapsto\sum_{i=1}^n \log c_\theta( \hat F_{n,1}(X_{i1}),
\ldots,
\hat F_{n,p}(X_{ip}) )
$.

\begin{theorem}
\label{thm:PLEefficiency}
Suppose that the parametrization $\theta\mapsto R(\theta)$ satisfies
Assumption~\ref{ass:parametrization}. The pseudo-likelihood estimator
$\hat{\theta}_n^{\mathrm{PLE}}$ is
semiparametrically efficient at $\prs
_{\theta,F_1,\ldots,F_p} \in\model$ in the sense of \eqref{eqn:efficiency} if
and only if, for every $m=1,\ldots,k$, the matrix
%
\begin{equation}
\label{eq:condPLEefficiency} L_m(\theta) - \tfrac{1}{2} \bigl( \diag
\bigl(L_m(\theta)\bigr) R(\theta) + R(\theta) \diag
\bigl(L_m(\theta)\bigr) \bigr) 
\end{equation}
with
\[
L_m(\theta) = R(\theta) \diag \bigl(\dot{R}_m(\theta)
S(\theta) \bigr) R(\theta)
\]
belongs to the linear span of $\dot{R}_1(\theta),\ldots,\dot
{R}_k(\theta)$.
\end{theorem}

The proof of Theorem~\ref{thm:PLEefficiency} is given in Appendix C %
in the supplement, after the proofs of the other results in this
section, on which it depends. Theorem~\ref{thm:PLEefficiency}
immediately implies that the pseudo-likelihood estimator is efficient
in the unrestricted model, in which each of the $p(p-1)/2$ off-diagonal
entries of the correlation matrix is a parameter. Indeed, in the
unrestricted model, the matrices $\dot{R}_1(\theta),\ldots,\dot
{R}_k(\theta)$, with $k = p(p-1)/2$, span the space of symmetric
matrices with zero diagonal, to which the matrix in \eqref
{eq:condPLEefficiency} clearly belongs. In Section~\ref{SEC:EXAMPLES},
efficiency of the pseudo-likelihood estimator will also be established
for exchangeable correlation matrices (Example~\ref{vb:exchangeable})
and for the versatile class of factor models (Example~\ref{ex:FactorModels}).

For the present section, it is convenient to assume that all marginal
distributions are uniform, so that the relevant nonparametric part of
the tangent space becomes $\mathcal{T}_{\prs_{\theta}}$. We maintain
the notation
$\mathbf{Z} = (Z_1,\ldots, Z_p)'$ with $Z_j = \Phi^{-1}(U_j)$ for
$j=1,\ldots,p$, where $\mathbf{U}$ is the identity mapping on the space $(0, 1)^p$,
which is equipped with the probability measure $\prs_\theta$ induced by
the Gaussian copula $C_\theta$. The distribution of $\mathbf{Z}$ is then
$N_p( 0, R(\theta) )$. Recall that $\operatorname{Sym}(p)$ is the
space of real,
symmetric $p \times p$ matrices.

\subsection{Quadratic functions and the tangent space}

For $A \in\operatorname{Sym}(p)$, define the function $\qdell_A \in
\rL_2^0(\prs
_\theta
)$ via
%
\begin{equation}
\label{eq:quadInfFunc} \qdell_A(\mathbf{U}) = \tfrac{1}2 \bigl(
\mathbf{Z}^\prime A \mathbf{Z} - \expecs_\theta\bigl[
\mathbf{Z}^\prime A \mathbf{Z}\bigr] \bigr).
\end{equation}
We call a score or influence function of this form ``quadratic'' and
``generated by $A$.'' For instance, the parametric score \eqref
{eqn:score} for $\theta_m$ is generated by $A = -\dot{S}_m(\theta)$,
while the semiparametrically efficient score \eqref{eqn:effscorem} is
generated by $A = D_\theta(\mathbf{g}_m(\theta)) - \dot{S}_m(\theta)$ with
$D_\theta( \cdot )$ defined in \eqref{eqn:Dmatrix}. By
Theorem~\ref
{THM:EFFICIENCY}, the one-step estimator has a quadratic influence
function, and in the proof of Theorem~\ref{thm:PLEefficiency}, we will
see that the pseudo-likelihood estimator has a quadratic influence
function, too.


First we characterize the matrices $A \in\operatorname{Sym}(p)$ that generate
quadratic forms $\qdell_A(\mathbf{U})$ belonging to $\mathcal{T}_{\prs
_{\theta}}$, the
nonparametric part of the tangent space.

\begin{lemma}
\label{lem:ATS}
Suppose that the parametrization $\theta\mapsto R(\theta)$ satisfies
Assumption~\ref{ass:parametrization}. The matrix $A \in\operatorname{Sym}(p)$
generates an element $\qdell_{A}(\mathbf{U})$ of $\mathcal{T}_{\prs
_{\theta}}$ if and
only if there exists a vector $\mathbf{b} \in\reals^p$ such that $A =
D_\theta(\mathbf{b})$.
\end{lemma}

Lemma~\ref{lem:ATS} motivates us to define a linear subspace of
$\operatorname{Sym}(p)$,
\[
\dot{T}_\theta = \bigl\{ D_\theta(\mathbf{b}) 
 |
\mathbf{b} \in\reals^p \bigr\},
\]
whose elements generate quadratic scores in the nonparametric part,
$\mathcal{T}_{\prs_{\theta}}$, of the tangent space. It follows from
part A in the
proof of Proposition~\ref{prop:effscore} that the dimension of $\dot
{T}_\theta$ is $p$.

Recall that we had endowed the space $\operatorname{Sym}(p)$ with the
inner product
$\langle{\cdot},{\cdot}\rangle_\theta$ in~\eqref
{eq:innerProductDef}. In the present
notation, the inner product can be written as
%
\begin{equation}
\label{eq:inpr:q} \langle{A},{B}\rangle_\theta= \cov_\theta \bigl(
\qdell_A(\mathbf{U}), \qdell_B(\mathbf{U}) \bigr),\qquad A, B
\in\operatorname{Sym}(p).
\end{equation}
It follows that the map $A \mapsto\qdell_A(\mathbf{U})$ constitutes an
isometry between $\operatorname{Sym}(p)$ and the linear subspace
$\mathcal{Q}_{\prs
_\theta} = \{ \qdell_A(\mathbf{U}) \mid A \in\operatorname{Sym}(p)\}$
of $\rL_2^0(
\prs
_\theta)$. Lemma~\ref{lem:ATS} then states that the intersection of
$\mathcal{Q}_{\prs_\theta}$ and $\mathcal{T}_{\prs_{\theta}}$ is
isometric to
$\dot
{T}_\theta$.

The construction of the one-step estimator \eqref{eqn:OSE} was based on
the efficient score function, which in turn was found through a
projection of the parametric score on (the orthocomplement of)
$\mathcal{T}_{\prs_{\theta}}$ (Proposition~\ref{prop:effscore}).
For quadratic
forms, the calculation of such projections is particularly simple.



\begin{corollary}
\label{lem:quadTangentSpaceOrth}
Suppose that the parametrization $\theta\mapsto R(\theta)$ satisfies
Assumption~\ref{ass:parametrization} and let $\theta\in\Theta$ and
$A\in\operatorname{Sym}(p)$. Statements \textup{(a)} and \textup{(b)} in Lemma~\ref
{prop:orthogonalitytangentspace} are both equivalent to:
\begin{longlist}[(c)]
\item[(c)] the matrix $A$ is orthogonal to $\dot{T}_\theta$.
\end{longlist}
%
\end{corollary}

\begin{corollary}
\label{cor:projection}
If the parametrization $\theta\mapsto R(\theta)$ satisfies
Assumption~\ref{ass:parametrization}, the projection of $A \in
\operatorname{Sym}(p)$
on $\dot{T}_\theta$ is equal to $D_\theta(\mathbf{b})$ and the projection
of $q_A(\mathbf{U})$ on $\mathcal{T}_{\prs_{\theta}}$ is equal to
$q_{D_\theta(\mathbf
{b})}(\mathbf{U})$, where $\mathbf{b} \in\reals^p$ is the unique solution of
%
\begin{equation}
\label{eq:quadProjDotT2} \diag \bigl( R(\theta) A - \diag(\mathbf{b}) - R(\theta) \diag(
\mathbf{b}) S(\theta) \bigr) = \mathbf{0}.
\end{equation}
\end{corollary}

Corollary~\ref{cor:projection} sheds new light on Proposition~\ref
{prop:effscore}. The projection of the parametric score $\Scorem( \mathbf
{U}; \theta) = \qdell_{- \dot{S}_m(\theta)}( \mathbf{U} )$ on
$\mathcal{T}_{\prs_{\theta}}$ is given by $\Scorem( \mathbf{U}; \theta
) - \dell_{\theta,m}^*(
\mathbf{U}; \theta) = \qdell_{ D_\theta( \mathbf{g}_m(\theta) )}(\mathbf
{U})$ with
the vector $\mathbf{g}_m(\theta)$ given by \eqref
{eqn:effscore_generator_weights}. But the latter equation says that
$-\mathbf{g}_m(\theta)$ is equal to the solution to \eqref
{eq:quadProjDotT2} with the matrix $A$ equal to $- \dot{S}_m(\theta)$.



Corollaries \ref{lem:quadTangentSpaceOrth} and \ref{cor:projection}
greatly simplify all tangent space projection calculations. Any
quadratic score or influence function $\qdell_A(\mathbf{U})$ generated by a
matrix $A$ such that $\diag( R(\theta) A ) = 0$ is automatically
orthogonal to the infinite-dimensional space $\mathcal{T}_{\prs
_{\theta}}$ in $\rL
_2^0( \prs_\theta)$. This property will be used extensively below in
the investigation of regularity and efficiency of estimators and of
adaptivity of the model.


\begin{remark}
By studying rank-based likelihoods, \citet{hoff:niu:wellner:12} conclude
that, with Gaussian marginals, the $p$ unknown marginal variances
generate the least-favorable directions. Indeed, using the calculations
in their Theorem~4.1, one may verify directly that these marginal
variances generate scores of the form $q_{D_\theta(\mathbf{b})}(\mathbf{U})$
with $\mathbf{b} \in\reals^p$.
\end{remark}

\subsection{Regularity}\label{sec:regularity}

The study of (semi)parametric efficiency is usually limited to \emph
{regular} estimators, which are estimators with the property that their
limiting distribution, after proper centering and rescaling, is the
same under any sequence of local alternatives [\citet{vdvaart:00}, Section~23.5,
page~365]. Consider an estimator $\hat{\theta}_n$ whose
components are asymptotically linear with quadratic influence
functions, that is, for all $\theta\in\Theta$ and all $m = 1,\ldots,
k$ there exists $A_m(\theta) \in\operatorname{Sym}(p)$ such that
%
\begin{equation}
\label{eq:aslin} \sqrt{n} ( \hat{\theta}_{n,m} - \theta_m )
= \frac{1}{\sqrt{n}} \sum_{i=1}^n
\qdell_{A_m(\theta)}( \mathbf{U}_i ) + o(1; \pr_\theta).
\end{equation}
(Recall that we focus on rank-based estimators, so that we may, without
loss of generality, assume that the margins are uniform.)
For an asymptotically linear estimator, regularity is equivalent to the
statement that the orthogonal projection of its influence function on
the full tangent set is equal to the efficient influence function
[\citet
{bkrw:93}, Proposition~3.3.1]. Since the full tangent set \eqref
{eqn:fulltangentspace} of the Gaussian copula model is spanned by the
nonparametric part $\mathcal{T}_{\prs_{\theta}}$ and the parametric
scores $\Scorem
(\mathbf{U}; \theta)$, $m = 1, \ldots, k$, regularity of $\hat{\theta
}_n$ in
\eqref{eq:aslin} is equivalent to
%
\begin{eqnarray}
\label{eq:condRegularity1} &&\qdell_{A_m(\theta)}( \mathbf{U} ) \perp\mathcal{T}_{\prs_{\theta
}},
\\
\label{eq:condRegularity2}& &\cov_\theta \bigl( \qdell_{A_m(\theta)}( \mathbf{U} ),
\dot{\ell }_{\theta,m'}( \mathbf{U} ) \bigr) = \delta_{m = m'},
\end{eqnarray}
for all $m, m' \in\{1,\ldots, k\}$. These equations pose restrictions
on $A_m(\theta)$, characterized by the following proposition.

\begin{proposition}
\label{prop:quadCondRegularity}
Suppose that the parametrization $\theta\mapsto R(\theta)$ satisfies
Assumption~\ref{ass:parametrization}. Let $\hat{\theta}_n$ be an
estimator sequence satisfying \eqref{eq:aslin} for every $\theta\in
\Theta$ and $m = 1,\ldots, k$. Then $\hat{\theta}_{n}$ is regular at
$\prs_\theta$ if and only if, for all $m, m' \in\{1,\ldots, k\}$,
%
\begin{eqnarray}
\label{eq:quadRegularityDef1} &&A_m(\theta) \perp\dot{T}_\theta,
\\
\label{eq:quadRegularityDef2} &&\bigl\langle{A_m(\theta)},{-\dot{S}_{m'}(
\theta)}\bigr\rangle_\theta = \delta_{m=m'},
\end{eqnarray}
or equivalently, if and only if, for all $m, m' \in\{1,\ldots, k\}$,
%
\begin{eqnarray}
\label{eq:quadCondRegularity1} \diag \bigl( R(\theta) A_m(\theta) \bigr) &=&
\mathbf{0},
\\
\label{eq:quadCondRegularity2} \tr \bigl( A_m(\theta) \dot{R}_{m'}(
\theta) \bigr) &=& 2 \delta_{m=m'}.
\end{eqnarray}
\end{proposition}


A matrix $A_m(\theta)$ satisfying the two conditions \eqref
{eq:quadCondRegularity1} and \eqref{eq:quadCondRegularity2} is called a
\emph{regular influence matrix} for $\theta_m$ at $\theta$. For the
one-step estimator, for instance, the influence matrices are
%
\begin{equation}
\label{eq:Aose} A_m^*(\theta) = \sum_{m'=1}^k
\bigl(I^{*-1}(\theta)\bigr)_{mm'} \bigl( D\bigl(
\mathbf{g}_{m'}(\theta) \bigr) - \dot{S}_{m'}(\theta)
\bigr),
\end{equation}
which can be seen from the right-hand side of \eqref{eqn:efficiency}
and the expression for the efficient score in Proposition~\ref{prop:effscore}.
By construction, the matrices $A_m^*(\theta)$ in \eqref{eq:Aose} are
regular influence matrices. Of course, this property also follows more
generally from the above description of regularity and the fact that
the one-step estimator is asymptotically linear in the efficient
influence function. In the proof of Theorem~\ref{thm:PLEefficiency}, we
check that the pseudo-likelihood estimator is regular, too.






For each $m = 1,\ldots, k$ and $\theta\in\Theta$, equations \eqref
{eq:quadRegularityDef1} and \eqref{eq:quadRegularityDef2} pose $p + k$
independent linear restrictions on $A_m(\theta)$. Indeed, the dimension
of $\dot{T}_p(\theta)$ is $p$ while the matrices $-\dot{S}_1(\theta),\ldots, -\dot{S}_m(\theta)$ are linearly independent and are not
contained in $\dot{T}_p(\theta)$; see part A of the proof of
Proposition~\ref{prop:effscore}. It follows that the set of regular
influence matrices for a given component of $\theta$ is an affine
subspace of $\operatorname{Sym}(p)$ of dimension $p(p+1)/2 - (p+k) =
p(p-1)/2 - k$.


\subsection{Efficiency}
\label{sec:efficiency}

In the unrestricted model, each pairwise correlation being a parameter,
there are $k = p(p-1)/2$ parameters, so that there is a unique regular
influence matrix for each component $\theta_m$, which then must be
equal to the efficient one, $A_m^*(\theta)$. In Example~\ref
{ex:quadRankCorr}, this matrix will be identified with the one
generating the influence function of the rank correlation estimator,
proving efficiency of the latter.

In structured Gaussian copula models, that is, those models where the
matrices $\dot{S}_m(\theta)$, $m=1,\ldots,k$, span a subspace of
dimension $k = \dim( \Theta)$ less than $p(p-1)/2$, multiple regular
quadratic influence functions exist.
Within this set of regular influence matrices, the efficient influence
matrices $A_m^*(\theta)$ admit the following characterization.

\begin{proposition}
\label{prop:quadCondEfficiency}
Suppose that the parametrization $\theta\mapsto R(\theta)$ satisfies
Assumption~\ref{ass:parametrization}. Consider a regular, rank-based
estimator $\hat\theta_n$ for which there exist matrices $B_1(\theta),\ldots, B_k(\theta) \in\operatorname{Sym}(p)$ such that the
influence functions of all
components $\hat{\theta}_{n,1},\ldots, \hat{\theta}_{n,k}$ are
quadratic and are given by matrices 
belonging to the linear span of $B_1(\theta),\ldots, B_k(\theta)$.
Write $B_{m,R}(\theta) = R(\theta) B_m(\theta) R(\theta)$. Then
$\hat
\theta_n$ is efficient at $\prs_\theta$ in the sense of \eqref
{eqn:efficiency} if and only if, for each $m$, the matrix
\begin{equation}
\label{eq:quadGeneralEfficiencyCriterionR} B_{m,R}(\theta) - \tfrac{1}{2} \bigl( \diag\bigl(
B_{m,R}(\theta) \bigr) R(\theta) + R(\theta) \diag\bigl(
B_{m,R}(\theta) \bigr) \bigr)
\end{equation}
belongs to the linear span of $\dot{R}_1(\theta),\ldots,\dot
{R}_k(\theta)$.
\end{proposition}

The characterization is based essentially on the fact that, by
Proposition~\ref{prop:quadCondRegularity}, the efficient influence
matrices $A_m^*(\theta)$ are the only regular influence matrices that
belong to the space
%
\begin{equation}
\label{eq:TangentSpace} \operatorname{span} \bigl( \dot{T}_\theta\cup\bigl\{
\dot{S}_1(\theta),\ldots, \dot{S}_k(\theta ) \bigr\}
\bigr).
\end{equation}
Moreover, the projection of any other regular influence matrix
$A_m(\theta)$ on the space~\eqref{eq:TangentSpace} is equal to
$A_m^*
(\theta)$.
The efficiency criterion for the pseudo-likelihood estimator in
Theorem~\ref{thm:PLEefficiency} is essentially a particular case of
Proposition~\ref{prop:quadCondEfficiency}.

\subsection{Adaptivity}
\label{subsec:adaptivity}

If the efficient information matrix $I^*(\theta)$ is equal to the
Fisher information matrix $I(\theta)$ in the parametric model with
known margins, the fact of not knowing the margins does not make a
difference asymptotically for the efficient estimation of $\theta$. For
semiparametric Gaussian copula models, there is a simple criterion for
the occurrence of this phenomenon, called adaptivity.

\begin{proposition}
\label{prop:quadCondAdaptivity}
The semiparametric Gaussian copula model is adaptive at $\prs_{\theta,F_1,\ldots,F_p} \in\model$ if and only if
%
\begin{equation}
\label{eq:quadCondAdaptivity} \diag \bigl( R(\theta) \dot{S}_m(\theta) \bigr)
 = 0,\qquad
m=1,\ldots,k.
\end{equation}
\end{proposition}

Obviously, adaptivity always occurs at the independence copula
(assuming it belongs to the model) as in that case $R(\theta)$ and
$S(\theta)$ equal the identity matrix and $\diag\dot{S}_m(\theta) = -
\diag\dot{R}_m(\theta) = 0$. See Example~\ref{vb:adaptivity} for a
model which is adaptive at a copula different from the independence
one. Still, adaptivity is the exception rather than the rule: most of
the time, not knowing the margins makes inference on the copula
parameter $\theta$ more difficult.


\section{Examples and simulations}
\label{SEC:EXAMPLES}

This section presents some analytical and numerical results for the
one-step and pseudo-likelihood estimators for a number of correlation
structures. For some cases, the pseudo-likelihood estimator is
efficient (unrestricted model, exchangeable model, factor model,
Toeplitz model in $p = 3$), sometimes it is almost efficient (circular
model), and sometimes it is quite inefficient (Toeplitz model in $p =
4$). The results of a Monte Carlo study indicate that the asymptotic
approximations to the finite-sample distributions are excellent for the
one-step estimator
. Adaptivity almost never occurs, except at independence and at a
contrived example. The simulation study is implemented in MATLAB
2012a and the code is available upon request. In the simulations, the
pseudo-likelihood estimator is used as (rank-based and $\sqrt
{n}$-consistent) pilot estimator.




\begin{example}[(Unrestricted model)]
\label{ex:quadRankCorr}
In the full, unrestricted model, there are $k=p(p-1)/2$ parameters,
which can be identified with the correlations $r_{ij}$ between $Z_i$
and $Z_j$ for $i=2,\ldots,p$ and $j=1,\ldots,i-1$. Efficiency of the
normal scores rank correlation estimator for the unrestricted Gaussian
copula model was already observed in \citet{klaassen:wellner:97}. Here
we obtain this result within our general algebraic analysis of
(possibly) structured Gaussian copula models.

First, one can check that for arbitrary Gaussian copula models, the
pseudo-score equations $\sum_{i=1}^n \Scorem( \hat{F}_{n,1}( X_{i1} ),\ldots, \hat{F}_{n,p}( X_{ip} ); \theta) = 0$ are equivalent to
%
\begin{equation}
\label{eq:pseudoScoreEq} \tr \bigl( \dot{S}_m(\theta) \bigl( R(\theta) -
\hat{R}_n \bigr) \bigr) = 0,\qquad m = 1,\ldots, k,
\end{equation}
where $\hat{R}_n$ is the $p \times p$ matrix
%
\begin{equation}
\label{eq:hatRn} \hat{R}_n = \frac{1}{n} \sum
_{i=1}^n \hat{\mathbf{Z}}_{n,i} \hat{
\mathbf{Z}}_{n,i}^\prime
\end{equation}
and $\hat{Z}_{n,ij} = \Phi^{-1}( \hat{F}_{n,j}(X_{ij}) )$ for $i = 1,\ldots, n$ and $j = 1,\ldots, p$. It follows that for the unrestricted
model, the pseudo-likelihood estimator is given by $\hat{R}_n$ itself.
The normal scores rank correlation estimator is just $\hat{R}_n /
\sigma
_n^2$ with $\sigma_n^2 = n^{-1} \sum_{i=1}^n [\Phi^{-1}(i/(n+1))]^{-2}
= 1
+ O(n^{-1} \log n)$, and is therefore asymptotically equivalent to the
pseudo-likelihood estimator. But the latter was already shown to be
efficient in the paragraph following Theorem~\ref{thm:PLEefficiency}.

The fact that for the unrestricted model, the normal scores rank
correlation estimator, the pseudo-likelihood estimator and the one-step
estimator are all asymptotically equivalent is due to the uniqueness of
the regular influence matrices for the correlation parameters. Indeed,
fix $i = 2,\ldots, p$ and $j = 1,\ldots, i-1$ and consider a regular
influence matrix $A_{ij}\in\operatorname{Sym}(p)$ for $r_{ij}$. The
derivative matrix
$\dot{R}_{ij}$ of $R$ with respect to $r_{ij}$ equals the $p\times p$
matrix whose elements are all zero, except for the $(i,j)$th and
$(j,i)$th elements, which equal unity. In view of \eqref
{eq:quadCondRegularity2}, regularity implies, for every $i' = 2,\ldots, p$ and $j' = 1,\ldots, i'-1$,
\[
2 \delta_{i=i',   j=j'} = \tr ( A_{ij} \dot{R}_{i'j'} ) = 2
(A_{ij} )_{i'j'}.
\]
Hence, the regularity conditions \eqref{eq:quadCondRegularity2} imply
that all off-diagonal elements of $A_{ij}$ are zero, but for $
(A_{ij} )_{ij}= (A_{ij} )_{ji} = 1$. Subsequently,
\eqref
{eq:quadCondRegularity1} implies $ (A_{ij} )_{ii}=
(A_{ij} )_{jj}=-r_{ij}$, while the other diagonal elements are
zero. Consequently, in the unrestricted model, there is only one
quadratic regular influence function for each pair $(i, j)$, which,
therefore, must be equal to the efficient influence function in~\eqref
{eq:effInfFun}.
\end{example}

\begin{example}[(Toeplitz model)]
\label{vb:toeplitz}
A Toeplitz model for the correlation matrix arises if there are $k =
p-1$ parameters $\theta_1,\ldots, \theta_{p-1}$ such that $r_{ij} =
\theta_{|i-j|}$, where $r_{ij}$ is the correlation between $Z_i$ and
$Z_j$. In dimension $p = 3$, for instance, the model is
$R_{12}(\theta_1,\theta_2) = R_{23}(\theta_1,\theta_2) = \theta_1$ and
$R_{13}(\theta_1,\theta_2) = \theta_2$,
and a brute-force calculation using the computer algebra system \textsf
{Maxima} (version 5.29.1, \url{http://maxima.sourceforge.net}) shows that the
inverse of the efficient information matrix is equal to the asymptotic
covariance matrix of the pseudo-likelihood estimator, calculated from
(C.1)--(C.2). 
The explicit formula for $I^{*-1}(\theta)$ is given in~(D.2). %

\begin{figure}[t!]

\includegraphics{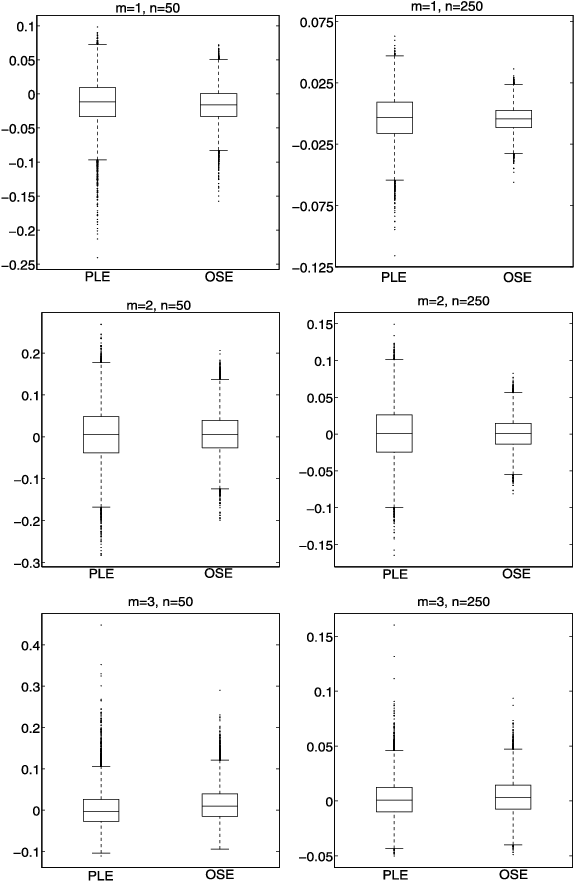}

\caption{Simulation results, based on 15,000 replications, for the
Toeplitz model in Example \protect\ref{vb:toeplitz} in dimension
$p=4$ at $\theta=(0.4945, -0.4593, -0.8462)^\prime$ and sample size
$n = 50$ (left) and $n = 250$ (right).
Boxplots of $\hat{\theta}_{n,m}^{\mathrm
{PLE}}-\theta_m$ and $\OSE[,m]-\theta_m$
for $m = 1$ (top), $m = 2$ (middle) and $m = 3$
(bottom).}\label{fig:toeplitz}\vspace*{-3pt}
\end{figure}

In dimension $p = 4$, however, the pseudo-likelihood estimator is no
longer efficient. The analytic expressions for the asymptotic variance
matrices of the one-step and pseudo-likelihood estimators are too long
to display, but for specific values of $\theta= (\theta_1, \theta_2,
\theta_3)^\prime$, they can easily be computed numerically. This
allows us to
search for values for $\theta$ in which the asymptotic relative
efficiency of the pseudo-likelihood estimator is particularly low. At
$\theta_* = (0.4945460, -0.4592764, -0.8462492)^\prime$, for instance,
the asymptotic relative efficiencies of the pseudo-likelihood estimator
with respect to the information bound are equal to $(18.3\%, 19.8\%,
96.9\%)$.
The finite-sample performance of the one-step and pseudo-likelihood
estimators were compared using 15,000 Monte Carlo samples of sizes $n =
50$ and $n = 250$. The boxplots of the estimation errors, shown in
Figure~\ref{fig:toeplitz}, confirm that for the components $\theta_1$
and $\theta_2$, the pseudo-likelihood estimator is quite inefficient at
$\theta= \theta_*$.
\end{example}

\begin{example}[(Exchangeable model)]
\label{vb:exchangeable}
The exchangeable Gaussian copula model is a one-parameter model in
which all off-diagonal entries of the $p \times p$ matrix $R(\theta)$
are equal to the same value of $\theta$ between $-1/(p-1)$ and $1$.
Efficiency of the pseudo-likelihood estimator for $\theta$ in dimension
$p = 4$ was
established in \citet{hoff:niu:wellner:12} and can, for general $p$, be
verified using
Theorem~\ref{thm:PLEefficiency}; see Appendix D 
in the supplement for some algebraic details. Using the computer
algebra system \textsf{Maxima}, we calculated the optimal asymptotic
variance for regular estimators of $\theta$ in dimensions three and four:
\[
I^{*-1}(\theta){} = \cases{ \frac{1}{3} (\theta-
1)^2 (2 \theta+ 1)^2, &\quad $ \mbox{if $p = 3$}$, \vspace*{2pt}
\cr
\frac{1}{6} (\theta- 1)^2 (3 \theta+ 1)^2, &\quad $
\mbox{if $p = 4$}$.}
\]
In contrast, if the margins are known, the optimal asymptotic variance
reduces to
\[
I^{-1}(\theta) = \cases{ I^{*-1}(\theta){} / \bigl(1 + 2
\theta^2\bigr), & \quad $\mbox{if $p = 3$}$, \vspace *{2pt}
\cr
I^{*-1}(\theta){} / \bigl(1 + 3 \theta^2\bigr), & \quad $\mbox{if
$p = 4$}$,}
\]
so that adaptivity occurs at independence $(\theta= 0)$ only.

We assessed the finite-sample performance of the one-step and pseudo-likelihood
estimators by 15,000 Monte Carlo samples of sizes $n = 50$ and $n =
250$ in dimension
$p = 3$ for $\theta$ in a grid of values between $-1/2$ and $1$; see
Figure~\ref{fig:exchangeable}\vadjust{\goodbreak} and Figure E.1 %
in the supplement. Even for $n = 50$, the finite-sample variance of the one-step
estimator is well approximated by its limit. For the pseudo-likelihood
estimator, the convergence is slower and its variance in finite samples
is a bit larger. The biases of both estimators are of comparable order
and are generally negligible relative to the variances.

\begin{figure}
\centering
\begin{tabular}{@{}c@{}}

\includegraphics{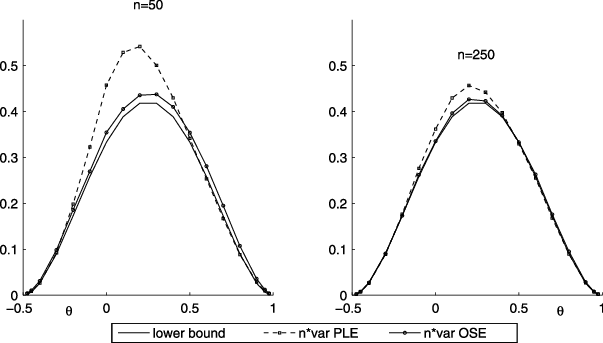}
\\
\footnotesize{(a) The semiparametric lower bound $I^{*-1}(\theta)$ and approximations
(based on 15,000 replications)}\\
\footnotesize{to $n \var_\theta( \hat\theta_{n}^{\mathrm{PLE}}
)$ and $n \var_\theta(\OSE)$
as a function of $\theta$.}\\[6pt]

\includegraphics{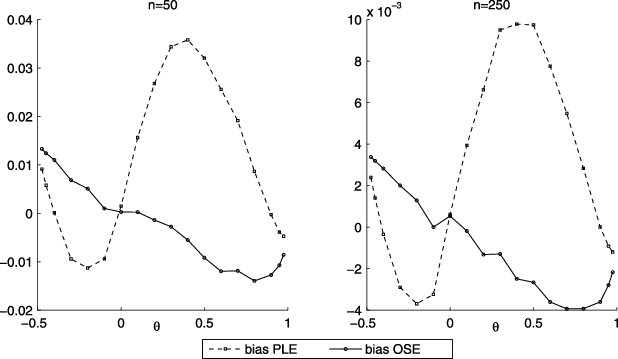}
\\
\footnotesize{(b) Approximations (based on 15,000 replications) to the biases $\expec_
\theta[\hat\theta_{n}^{\mathrm{PLE}}] - \theta$ and}\\
\footnotesize{$ \expec_\theta[\OSE] -
\theta$
as a function of $\theta$.}
\end{tabular}
\caption{Simulation results for the exchangeable model of
Example \protect\ref{vb:exchangeable} in dimension $p=3$ for
$\theta\in\{ -0.475, -0.45 \} \cup\{ k/10 \mid k=-4,\ldots, 9 \}
\cup\{ 0.95, 0.975 \}$ and $n \in\{ 50, 250
\}$.}\label{fig:exchangeable}\vspace*{-6pt}
\end{figure}

To assess the impact of the dimension, we also compare the one-step and
pseudo-likelihood
estimators in dimension $p = 100$ for 15,000 Monte Carlo samples of
size $n = 50$ at
$\theta= 0.25$. Boxplots of the estimation errors are shown in
Figure E.2 
in the supplement. Although the variances of both estimators are about
the same, the pseudo-likelihood estimator suffers from a large bias,
whereas the one-step estimator remains centered around the true value.
\end{example}

\begin{example}[(Circular model)]
\label{vb:circular}
In \citet{hoff:niu:wellner:12}, the four-dimensional circular model
\[
R(\theta) = \pmatrix{ 1 & \theta& \theta^2 & \theta\vspace*{2pt}
\cr
\theta& 1 & \theta& \theta^2 \vspace*{2pt}
\cr
\theta^2 &
\theta& 1 & \theta\vspace*{2pt}
\cr
\theta& \theta^2 & \theta& 1}, \qquad -1 <
\theta< 1,
\]
is introduced as a one-parameter Gaussian copula model where the
pseudo-likelihood estimator is not efficient. The optimal asymptotic
variance if margins are unknown, the asymptotic variance of the
pseudo-likelihood estimator, and the inverse Fisher information if
margins are known can be computed explicitly:
\begin{eqnarray*}
I^{-1}(\theta) &=& \frac{I^{*-1}(\theta)}{1 + 2 \theta^2} \le I^{*-1}(\theta)=
\frac{1}{4} \bigl(1 - \theta^2\bigr)^2
\\
&\le& I^{*-1}(\theta) \biggl( 1 + \frac{2 \theta^6}{(1 + 2\theta^2)^2} \biggr) =
\sigma^2_{\mathrm{PLE}}.
\end{eqnarray*}
%
Even though the pseudo-likelihood estimator is not efficient, its
asymptotic relative efficiency is close to $100\%$, except for $\theta$
close to $1$ or $-1$. Adaptivity occurs at independence ($\theta= 0$) only.

We assessed the finite-sample performance of the one-step and
pseudo-likelihood estimators by 15,000 Monte Carlo samples of sizes $n
= 50$ and $n = 250$ for $\theta$ in a grid of values between $-1$
and $1$. The results are comparable to those for the exchangeable
model: see Figures E.3--E.4
in the supplement.
\end{example}

\begin{example}[(Factor models)]
\label{ex:FactorModels}
Factor models are a popular tool for dimension reduction. In dimension
$p \ge3$, set
%
\begin{equation}
\label{eq:FactorModel} \mathbf{Z} = \theta\mathbf{W} + \bigl(\Id-\diag \bigl(\theta
\theta^\prime \bigr) \bigr)^{1/2} \bolds{\varepsilon},
\end{equation}
where $\theta$ denotes a $p \times q$ matrix, $q < p$, and where $\mathbf
{W}$ and $\bolds{\eps}$ are independent random vectors of dimensions $q$
and $p$, respectively, such that $(\mathbf{W}', \bolds{\eps}')' \sim N_{q+p}(
\mathbf{0}, I_{p+q})$. The parameter space $\Theta$ is an open subset of
$\{ \theta\in\reals^{p \times q} \mid(\theta\theta^\prime)_{jj}
< 1,
j=1,\ldots,p \}$; in particular, $\Id-\diag (\theta\theta
^\prime
)$ is a diagonal matrix with positive elements on the diagonal. As the
variance matrix of $\mathbf{Z}$ is a correlation matrix, \eqref
{eq:FactorModel}~defines a Gaussian copula model with
%
\begin{equation}
\label{eq:FactorModel:R} R(\theta) = \theta\theta^\prime+ \bigl(\Id-\diag \bigl(
\theta\theta^\prime \bigr) \bigr) = \Id+ \diagrem \bigl(\theta
\theta^\prime \bigr)
\end{equation}
in terms of the diagonal-removal operator $\diagrem(A) = A - \diag(A)$
for $A \in
\reals^{p \times p}$. The parameter $\theta$ is not identifiable: if
$O$ is an
orthogonal $q \times q$ matrix, then $R(\theta O) = R(\theta)$. Still, by
Corollary D.2 
in Appendix D, 
we may study efficiency of the pseudo-likelihood estimator for the parameter
$\nu$ in any reparametrization $\nu\mapsto\theta(\nu)$ satisfying
Assumption D.1
, which makes the model identified, by the criterion in Theorem~\ref
{thm:PLEefficiency}. After some calculations, which are detailed in the
Appendix, the criterion can be shown to be satisfied, confirming the
efficiency of the pseudo-likelihood estimator for Gaussian factor
copula models.
\end{example}

\begin{example}[(Adaptivity)]
\label{vb:adaptivity}
Proposition~\ref{prop:quadCondAdaptivity} gives a necessary and
sufficient criterion for adaptivity of a Gaussian copula model at a
certain value of $\theta$. Adaptivity always occurs at the independence
copula but, apart from this, is the exception rather than the rule.
With some trial and error, other (artificial) examples can be
constructed. For instance, the one-parameter model in dimension $p = 3$
given by $R_{12}(\theta) = R_{13}(\theta) = \theta^2 + 0.5$ and
$R_{23}(\theta) = \theta+ 0.25$,
for $\theta$ in a neighborhood of $0$, can be verified to be adaptive
at $\theta= 0$.
\end{example}

\section{Conclusion and extensions}
\label{sec:conclusions}

The present paper provides a semiparametrically efficient, rank-based
estimator for the copula parameter in structured Gaussian copula models
under mild
conditions on the parametrization
$\theta\mapsto R(\theta)$ of the correlation matrix. This gives a
positive answer to the conjecture formulated in \citet{hoff:niu:wellner:12} that in Gaussian copula models,
semiparametrically efficient, rank-based estimators do exist. The
estimator is based on the analysis of the tangent space structure of
the model and the explicit calculation of the efficient score function.
Simulations indicate that the large-sample distribution provides an
accurate approximation to the finite-sample distribution of the
estimator, even in large dimensions. 

Moreover, we show that inference in structured Gaussian copula models
can be studied using a convenient finite-dimensional algebraic
representation of relevant scores and influence functions. This leads
to straightforward conditions to verify the regularity or the
efficiency of existing estimators. In particular, we provide a
convenient necessary and sufficient condition for the semiparametric
efficiency of the pseudo-likelihood estimator. It follows, for example,
that this estimator is efficient in models that exhibit a suitable
factor structure. However, we also provide examples where its relative
efficiency can be as low as $20\%$. Several other concrete examples
complement the analysis.

\subsection{Other copulas}
\label{ss:othercopulas}

A natural question is how to perform rank-based, semiparametrically
efficient inference in general semiparametric copula models.
The Bickel--Le Cam technique of a one-step update based on a suitable
pilot estimate should work in general, too. However, our derivation of
the efficient score function and information matrix (Proposition~\ref
{prop:effscore}) and the proof of the asymptotic normality of the
estimator (Theorem~\ref{THM:EFFICIENCY}) heavily relied on the
Gaussian-copula assumption, by which all relevant score and influence
functions are quadratic forms in the Gaussianized observations. In
general, one will have to pass through the Sturm--Liouville equations
in (A.11) 
in the supplement.

Another point of concern is what happens if the model is misspecified.
Suppose, for instance, that the true copula is elliptical with a
correlation matrix $R(\theta)$. Then the Gaussian quasi-score $\Score$
in \eqref{eqn:score} will still be centered, and statistical inference
procedures derived from it can be expected to be $\sqrt{n}$-consistent.
Semiparametric efficiency will be lost, however, since the structure of
the tangent space will have changed.

\subsection{Regression models}
\label{SS:REGRESSION}

An important application of (Gaussian) copula models lies in joint
regression analyses [\citet{song:2000,song:li:yuan:2008,masarotto:varin:2012}]. Consider a generalized linear model for each
component of a vector of dependent variables $\mathbf{Y} = (Y_1,\ldots,
Y_p)$. The vector of covariates, $\mathbf{X}$, is common to each of the $p$
model equations. Assume that the joint distribution of the vector of
error variables $(\eps_1,\ldots, \eps_p)$ in the $p$ model equations
has copula $C_\theta$ with unknown parameter vector $\theta$. The joint
conditional distribution of $\mathbf{Y}$ given $\mathbf{X}$ is then
parametrically specified. Estimating the $p$ vectors of regression
parameters jointly with $\theta$, for instance, by maximum likelihood,
is potentially more efficient than fitting each of the $p$ univariate
models separately.

The question is what happens in the semiparametric context, when the
marginal distributions of the errors $\eps_j$ are left unspecified,
except for a restriction to identify the model parameters. Does joint
modeling still lead to more efficient inference on the regression
coefficients? Is it still possible to estimate the (Gaussian) copula
parameter efficiently using ranks?

We illustrate the possibilities and difficulties by means of a simple
example. Consider a bivariate ``dependent'' variable $\mathbf{Y} =
(Y_1,Y_2)$ and a univariate ``explanatory'' variable $X$. The precise
setup is described by
\[
Y_j = \beta_j X + \alpha_j +
\eps_j, \qquad j \in\{ 1, 2 \},
\]
where $\varepsilon_j= F_j^{-1}( \Phi(Z_j) )$, for $j \in\{1, 2\}$,
with $(Z_1,Z_2)$ bivariate standard normal with correlation $\theta\in
(-1, 1)$ and $F_1,F_2 \in\Fac$ with absolutely continuous densities
$f_1, f_2$, respectively. The model parameters are identified by
imposing a location restriction on $F_1$ and $F_2$; for instance, their
means or their medians should be zero. We assume that $X$, with $0<\var
X<\infty$, is independent of $(Z_1, Z_2)$, has an unknown density $f_X$
w.r.t. some dominating measure, and is exogenous, that is, its
distribution does not depend on $(\alpha_1,\alpha_2,\beta_1,\beta
_2,\theta,F_1,F_2)$. Finally, we assume that $f_1$ and $f_2$ have
finite Fisher information for location, that is, $\int
[(f_j'(x))^2/f_j(x)]  \, \mathrm{d}x
< \infty$ for $j \in\{1, 2\}$.%

We are both interested in efficient inference about the copula
parameter $\theta$ and about the regression coefficients $\beta_1$ and
$\beta_2$, in the presence of the nuisance parameters $\alpha_1$,
$\alpha_2$, $f_1$, and $f_2$. We follow the tangent space calculations
as in Section~\ref{SEC:TANGENT}. Conditionally on $X$, the joint
distribution of $\mathbf{Y}$ has a bivariate Gaussian copula $C_\theta$
with correlation parameter $\theta$. The density of $(\mathbf{Y}, X)$ is
given by
\begin{eqnarray*}
(y_1, y_2, x)& \mapsto& c_\theta \bigl(
F_1( y_1 - \alpha_1 - \beta_1 x),
F_2( y_2 - \alpha_2 - \beta_2 x)
\bigr)
\\
&&{}\times f_1( y_1 - \alpha_1 -
\beta_1 x) f_2( y_2 - \alpha_2 -
\beta_2 x) f_X(x).
\end{eqnarray*}
%
The scores of the Euclidean parameters are given by
%
\begin{eqnarray}
\label{eq:regression:theta} \dot\ell_\theta(\mathbf{Y}, X) &=& \frac{\rd}{\rd\theta} \log
c_\theta\bigl( F_1( \varepsilon_1 ),
F_2( \varepsilon_2 ) \bigr),
\\
\label{eq:regression:alpha} \dot\ell_{\alpha_j}(\mathbf{Y}, X) &=& - \dot
\ell_j\bigl( F_1( \varepsilon_1 ),
F_2( \varepsilon_2 );\theta \bigr) f_j(
\varepsilon_1) - \frac{f_j^\prime}{f_j}(\varepsilon_1),
\\
\label{eq:regression:beta} \dot\ell_{\beta_j}(\mathbf{Y}, X) &=& X \dot
\ell_{\alpha_j}(\mathbf{Y}, X),
\end{eqnarray}
for $j \in\{1, 2\}$. Irrespective of the specific location restriction
that is used to identify $\alpha_1$ and $\alpha_2$, the set of
distribution functions $G_j$ that are obtained as those of $\alpha_j +
\varepsilon_j = Y_j - \beta_j X_j$ is unrestricted. As a consequence,
the tangent space generated by the nuisance parameters $\alpha_1$,
$\alpha_2$, $f_1$ and $f_2$ is equal to the collection of the score
functions $s(Y_1 - \beta_1 X_1,   Y_2 - \beta_2 X_2)$ with $s$ in
$\mathcal{T}_{\prs_{\theta,G_1,G_2}}$ as in \eqref{eq:tangentspace}.

We first consider efficient estimation of the copula parameter $\theta
$. As the score function for $\theta$ and the tangent space generated
by the nuisance parameters $\alpha_1$, $\alpha_2$, $f_1$, and $f_2$
are, up to isometry, identical to those in Section~\ref{SEC:TANGENT},
the efficient score for estimating $\theta$ in the presence of the
nuisance parameters $\alpha_1$, $\alpha_2$, $f_1$, and $f_2$ remains
$\dot\ell_\theta^{*}(F_1(\varepsilon_1),F_2(\varepsilon_2);\theta)$.
Given the independence of $X$ and $(\varepsilon_1,\varepsilon_2)$, the
function $\dot\ell_\theta^{*}(F_1(\varepsilon_1),F_2(\varepsilon
_2);\theta)$ is automatically orthogonal to $\dot\ell_{\beta_j}(\mathbf
{Y},X) = X   \dot\ell_{\alpha_j}(\mathbf{Y},X)$, for $j \in\{1,2\}$.

These tangent space calculations show the possibility of efficient
estimation of Gaussian copula parameters in joint regression analyses.
A formal proof of efficiency of the OSE would be significantly
complicated by the fact that one now has to rely on \emph{aligned
ranks}. That is, the ranks to be used for the computation of the
initial estimator and the update step would be those of the residuals
based on some initial estimates $\hat{\alpha}_j$ and $\hat{\beta}_j$,
for $j \in\{1,2\}$. Techniques for this exist [\citeauthor
{hallin:vermandele:werker:2006} (\citeyear
{hallin:vermandele:werker:2006}) and the references therein] but require
subtle analysis of remainder terms.

Usually, the interest lies in the estimation of the regression
parameters $\alpha_j$ and $\beta_j$ for $j \in\{1,2\}$. In order to
identify $\alpha_j$, an identification restriction on the location of
$f_j$ is needed. Focusing on rank-based procedures, a natural choice is
to use a median restriction, that is, to impose $\int_{-\infty}^0 f(x)
\,\rd x = 1/2$. In univariate settings, this problem has been studied
extensively and semiparametrically efficient inference procedures can
be based on signs and ranks [\citet{hallin:vermandele:werker:2006}].

Concerning the estimation of $\beta_j$, its efficient score is given by
the residual of the projection of $\dot\ell_{\beta_j}(\mathbf{Y},X)$ on the
tangent space generated by the score functions for $\alpha_1, \alpha_2,
f_1, f_2$ and $\beta_{3-j}$.
Elementary calculations (see Appendix F 
in the supplement) show that this efficient score is given by
%
\begin{equation}
\label{eq:effScoreBetaJoint} \dot\ell^*_{\beta_j}(Y,X) = (\dot\ell_{\alpha_j} -
c_j \dot \ell _{\alpha_{3-j}} ) ( X - \expec X ),
\end{equation}
with $c_j=\operatorname{cov}( \dot\ell_{\alpha_1},\dot\ell
_{\alpha
_2})/\operatorname{var}( \dot\ell_{\alpha_{3-j}} )$.

It is interesting to compare this efficient score to the univariate
regression case, that is, to the case where only $(Y_j,X)$ is observed.
It is known [see Example~3 in \citet{bickel:1982}] that in that case
the efficient score is given by
%
\begin{equation}
\label{eq:effScoreBetaUniv} 
-\frac{f_j^\prime}{f_j}(\varepsilon_j) (X -
\expec X).
\end{equation}
In general, the score functions in \eqref{eq:effScoreBetaJoint} and
\eqref{eq:effScoreBetaUniv} do not coincide, which suggests that
efficiency gains are possible from joint regression analyses. The
numerical size of the gains would have to be investigated further. One
interesting observation is that, at Gaussian marginals, the efficient
score functions in \eqref{eq:effScoreBetaJoint} and \eqref
{eq:effScoreBetaUniv} do coincide as, again, follows from
straightforward calculations detailed in the supplement. This result is
well-known from the literature on Seemingly Unrelated Regressions [see
Chapter~12 in \citet{Davidson:MacKinnon:2004}], but does not extend to
non-Gaussian distributions. Indeed, obtaining more efficient estimators
of regression parameters by joint analyses using copulas was exactly
the point in \citet{song:2000} and \citet{song:li:yuan:2008}.

\section*{Acknowledgments}

The authors wish to thank the referees for constructive suggestions and
for pointing out related literature, in particular, regarding the
extension to joint regression analyses. The authors gratefully
acknowledge extensive discussions with John H.~J. Einmahl (Tilburg
University) and Christian Genest (McGill University).

%
\begin{supplement}[id=suppA]
\stitle{Supplement to the paper: ``Semiparametric Gaussian copula models''}
\slink[doi]{10.1214/14-AOS1244SUPP} 
\sdatatype{.pdf}
\sfilename{aos1244\_supp.pdf}
\sdescription{The supplement contains the proofs for the results in
this paper as well as some additional figures for the Monte Carlo
simulations reported in Section~\ref{SEC:EXAMPLES}.}
\end{supplement}

%


\printaddresses
\end{document}